\documentclass[11pt,3p,times]{elsarticle}

\usepackage{amsmath,amsfonts,amssymb}
\usepackage{nicematrix}
\usepackage{scalerel}
\usepackage{amsthm}
\usepackage{comment}
\usepackage{mathtools}
\usepackage{xcolor}
\usepackage{tikz}
\usepackage{soul}

\usepackage{float}

\usepackage{xpatch}

\newcommand{\proofnameformat}{\itshape}
\xpatchcmd{\proof}{\itshape}{\proofnameformat}{}{}
\renewcommand{\proofnameformat}{\itshape\color{astral}}

\usepackage{amsthm,thmtools,xcolor}

\declaretheoremstyle[
  headfont=\color{astral}\normalfont\bfseries,
]{colored}

\declaretheorem[
  style=colored,
  name=Theorem,
]{thm}

\declaretheorem[
  style=colored,
  name=Lemma,
]{lem}

\declaretheorem[
  style=colored,
  name=Corollary,
]{coro}

\declaretheorem[
  style=colored,
  name=Remark,
]{rmk}

\newcommand{\R}{{\mathbb{R}}}

\newcommand{\N}{{\mathbb{N}}}
\renewcommand{\arraystretch}{1.75}
\usepackage{sectsty}

\definecolor{astral}{RGB}{46,116,181}
\subsectionfont{\color{astral}}
\sectionfont{\color{astral}}

\usepackage{amssymb}
\usepackage{amsthm}

\usepackage{fancyhdr}
\usepackage[figuresright]{rotating}

\usepackage{hyperref}
\hypersetup{
    colorlinks = false,
    linkbordercolor = {green}
}

\begin{document}

\begin{frontmatter}




\title{\color{astral}{\huge \textbf{Fixed points of Personalized PageRank centrality: \\ From irreducible to reducible networks}}}


\author{\color{black}{David Aleja$^{1, 2, 3}$, Julio Flores$^{1,2}$, Eva Primo$^{1, 2}$, Daniel Rodríguez$^{1, 2}$ and Miguel Romance$^{1, 2, 3}$}}

\address{$^{1)}$Departamento de Matemática Aplicada, Ciencia e Ingeniería de los Materiales y Tecnología Electrónica,
Universidad Rey Juan Carlos, 28933 Móstoles (Madrid), Spain\\
$^{2)}$Laboratory of Mathematical Computation on Complex Networks and their Applications, Universidad Rey Juan Carlos,
28933 Móstoles (Madrid), Spain\\
$^{3)}$Data, Complex networks and Cybersecurity Research Institute, Universidad Rey Juan Carlos, 28028 (Madrid), Spain}

\begin{abstract}
 In this paper we analyze the PageRank of a complex network as a function of its personalization vector. By using this approach, a complete characterization of  the existence and uniqueness of fixed points of PageRank of a graph is given in terms of the number and nature of its strongly connected components. The method presented  includes the use of a {\it feedback-PageRank} in order to compute exactly  the fixed points following the classic {\it Power's Method} in terms of the (left-hand) Perron vector of each strongly connected components.\\

\hspace*{-5mm}{\bf\color{astral}{Keywords:}} PageRank, fixed-point iterations, irreducible matrix, Perron-Frobenius theory, strongly connected network, Power Iteration Method\\
 
\hspace*{-5mm}{\bf\color{astral}{AMS subject classifications:}} 15A60, 47J26
 \end{abstract}




\end{frontmatter}


\section{Introduction}
\label{Section:introduction}

PageRank is one of the most relevant and successful examples of a centrality measure in complex networks that jumps from theory to real-life applications. Originally developed by Larry Page and Sergey Brin in the late 1990s to rank web pages in Google’s search engine \cite{BrinPage,PageBrin}, PageRank has since demonstrated an exceptional degree of ubiquity, being applied to a wide range of real-world problems beyond its original domain. Actual applications of PageRanks include problems in Social Sciences (for example measuring the influence of users on on-line platforms, where social interactions form complex interaction patterns, \cite{Riquelme}), the protein-protein interaction networks or metabolic pathways in Biology \cite{Ivan},  systemic risk in financial networks, identifying key institutions whose failure could cascade through the system in Economics \cite{Yun},  brain connectivity in Neuroscience \cite{Williamson}, the scientific impact of academic papers \cite{Senanayake} or public transportation modeling and analysis \cite{Criado}, among many others.

Roughly speaking, PageRank in a complex network $\mathcal{G}=(V,E)$ is the steady-state distribution of a random walk navigating on $\mathcal{G}$ such that at each step, the walker randomly chooses one of the available outgoing links from the current node, 
but with a given probability, he may {\it ``teleport"} to any other node in the network instead of following a link (see, for example \cite{Gleich,LM2006,PageBrin}). As we will see in more details in Section~\ref{Section:notation}, the basic ingredients for computing the PageRank in a complex network $\mathcal{G}=(V,E)$ are: 
\begin{itemize}
 \item the probability that the random walker follows a link rather than {\it``teleporting"} to a random node (the so called {\it damping factor}),
 \item a probability distribution vector on the nodes of the network that chooses a random node once the walker decides teleporting instead of using the network link's structure ({\it personalization vector}).  
\end{itemize}
There is a long list of excellent references that studied analytically the behavior of PageRank, as well as the influence and  sensibility of the damping factor and personalization vectors in PageRank, such as \cite{Bianchini, Boldi, Gleich, LM2004,LM2006}.

If we take a complex network $\mathcal{G}=(V,E)$ of $N\in\N$ nodes, and we fix a damping factor $\lambda\in (0,1)$, we observe that the basic features of PageRank allows for it to be mathematically understood as a function of the personalization vector \cite{Contreras}, that is, if we denote the $(N-1)$-simplex 
\[
\Delta_N^+=\left\{\bold{x}=(x_1,\cdots,x_N)^T\in\R^{N\times1}\,:\,\enspace x_1+\cdots+x_N=1,\, x_i> 0,\, {\text{for all }}1\le i\le N \,\right\},
\] 
then the PageRank is a function $PR_\lambda$ from $\Delta_N^+$ into itself. Hence, if we take a vector $\bold{v}\in\Delta_N^+$, then $PR_\lambda(\bold{v})$ is the PageRank of network $\mathcal{G}$ with damping factor $\lambda$ and personalization vector $\bold{v}$  (further details about the nature and behavior of such function will be given in Section~\ref{Section:notation}). 

This functional model approach to  PageRank suggests to employ the classical language of Functional Analysis and  consider functional properties of the operator $PR_\lambda$, such as the existence of fixed points, spectral properties and so on, as they may provide valuable information about the centrality of nodes of the underlying complex network $\mathcal{G}$. In this sense, the main goal of this paper is to analyze the existence and uniqueness, as well as the  computation, of the fixed points of $PR_\lambda$ in terms of the structure of $\mathcal{G}$. More precisely, we will prove that if  $\mathcal{G}$ is strongly connected, then there is always a unique fixed point of $PR_\lambda$ (with an analytical result that computes it precisely), while if $\mathcal{G}$ is not strongly connected, the existence of fixed points is proven and the uniqueness is related to the structure of the strongly connected  components of the graph (sinks). Note that a fixed point of $PR_\lambda$ can be understood as a vector $\bold{v}$, such that the PageRank of $\mathcal{G}$, with damping factor $\lambda$ and personalization vector $\bold{v}$, is $\bold{v}$ itself; in other words the teleportation vector considered gives exactly the PageRank. As we will see, the method used  includes the use of a {\it feedback-PageRank} in order to compute the fixed point - following the classic {\it Power's Method} for computing the fixed point of an operator. At this point, it is worth mentioning that some well-known results in Functional Analysis concerning fixed points (for instance, the Banach fixed point Theorem) no longer apply in this context, since the operator $PR_\lambda$ from $\Delta_N^+$ into itself is a non-contractive mapping as its spectral radius is equal to 1.

In addition to the intrinsic mathematical interest of analyzing the fixed points of $PR_\lambda$ as an operator,  the proposed methods and results are potentially useful in some applications. If we consider a temporal network (i.e. a complex network whose topological structure changes over time) and a random walker such that the time scale of the structure evolution is much slower that the time scale of the random walker itself, then it seems natural that the personalization vector could be modified along time. In this regard the equally natural  question arises: how can we select a personalization vector that is related to the structural properties of the network itself? In some applications, it may be natural to teleport to more central nodes with higher probability than to the less popular ones. Let us, for example, consider a web graph, where the random walkers (web-surfers) use the PageRank of the nodes in a previous day as the personalization vector to select the destination of a teleportation. By using this model, the steady-state of this stochastic process should be the  {\it feedback-PageRank}, if it exists, and therefore it should be the fixed point of the operator $PR_\lambda$, as  pointed out before.

The remainder of the paper is organized as follows. In Section \ref{Section:notation} we fix the notation and some preliminary definitions related to the PageRank algorithm. In Section \ref{Section:irreducible} we present some results concerning the row-stochasticity of the matrix $X(\lambda)$ and we state the main result of this paper on the convergence of the iterative process of the PageRank vector $\pi\in\R^{N\times 1}$ under the action of the matrix $X(\lambda)$, specifically for the case when the matrix $P_A$ is row-stochastic and irreducible. In Section \ref{Section:reducible} we extend this convergence to the case that $P_A$ is row-stochastic but no longer irreducible. To this aim, we investigate three basic cases: the diagonal case, the zero-block column case and the general reducible case. It is worth mentioning that not only we prove the convergence of this iterative process, but also explicitly obtain the form of the limit vector, which is expressed in terms of the left-hand Perron vector of the irreducible parts of the matrix $P_A$. Finally, in Section \ref{Section:dangling-cluster}, we apply the results obtained in the previous sections to show that the iterative process of the PageRank vector under the action of the matrix $X(\lambda)$ converges for networks that are not strongly connected. In fact, the limit vector of this iterative process is related to the left-hand Perron vector of the matrices corresponding to the dangling components of the network.
\section{Notation and preliminary definitions}
\label{Section:notation}

We recall some standard notation that will be used throughout the paper. Vectors of $\mathbb{R}^{N\times 1}$ will be denoted by column matrices and we will use the superscript $T$ to indicate matrix transposition. The vector of $\mathbb{R}^{N\times 1}$ with all its components equal to 1 will be denoted by $\mathbf{e}$, that is, $\mathbf{e} = (1,\dots,1)^T$ . A matrix $A=(a_{ij})_{i,j}$ will be called \em non-negative \em (respectively \em positive \em) if all its entries satisfy $a_{ij}\ge 0$ (respectively $a_{ij}>0$). The same applies to vectors when being column matrices.\\

Let $\mathcal{G}=(V,E)$ be a directed graph where $V=\{1,2,\dots,N\}$ is the set of nodes and $N\in\mathbb{N}$. The pair $(i,j)$ belongs to the set $E$ if and only if there exists a link connecting node $i$ to node $j$. The \textit{adjacency matrix} of $\mathcal{G}$ is an $N\times N$-matrix
\begin{equation*}
A=(a_{ij})_{i,j}\quad \text{ where }\quad
a_{ij}=
\begin{cases}
    1,\text{ if $(i,j)$ is a link of $\mathcal{G}$,}\\
    0,\text{ otherwise}.
\end{cases}
\end{equation*}

A link $(i,j)$ is said to be an \textit{outlink} for node $i$ and an \textit{inlink} for node $j$. We denote $k_{out}(i)$ the \textit{outdegree} of node $i$, i.e, the number of outlinks of a node $i$. Notice that $k_{out} (i) = \sum_{k}a_{ik}$. A node will be referred to as a \em source \em if it has no inlinks. Also, the graph $\mathcal{G}$ may have \textit{dangling nodes}, which are nodes $i\in V$ with zero outdegree. Dangling nodes are characterized by a vector $\mathbf{d}=(d_1,\cdots,d_N)^T\in\mathbb{R}^{N\times 1}$ with components $d_i$ defined by
$$d_{i}=
\begin{cases}
    1,\text{ if $i$ is a dangling node of $\mathcal{G}$,}\\
    0,\text{ otherwise}.
\end{cases}$$
Without loss of generality,  the directed graph $\mathcal{G}$ can be weighted with $a_{ij}\geq0$. We remark that the results in the paper remain valid if $\mathcal{G}$ is weighted.\\

At this point, let $P_A = (p_{ij})_{i,j}\in\mathbb{R}^{N\times N}$ be the row-stochastic matrix associated to directed graph $\mathcal{G}$ defined in the following way:
\begin{enumerate}
    \item If $i$ is a dangling node, then $p_{ij}=0$ for all $j=1,\dots,N$.
    \item Otherwise, $p_{ij}=a_{ij}/k_{out}(i)=a_{ij}/\sum_{k}a_{ik}$.
\end{enumerate}

Roughly speaking, each coefficient $p_{ij}$ is the probability of moving from the node $i$ to the node $j$. The matrix $P_A$ will be referred to as the \em row-normalization \em of $A$.  Note that if we take
\begin{equation}\label{eq:diagonal}
D=\left(
\begin{array}{cccc}
k_{out}(1)& 0 & \cdots & 0\\
0 & k_{out}(2) & \cdots & 0\\
\vdots & \cdots & \ddots & \vdots\\
0 & \cdots & \cdots & k_{out}(N)
\end{array}
\right),
\end{equation}
then the matrix $P_A$ is given by $P_A=D^{-1}A$.\\

One of the most remarkable features of the PageRank algorithm is that it contemplates the possibility to travel from one node to any other node in the graph. This teleportation probability is given by the \textit{personalization vector} $\textbf{v}\in\R^{N\times1}$ which, in fact, is a probability distribution vector. In addition, in the presence of dangling nodes, a probability vector $\bold{u}\in\mathbb{R}^{N\times1}$ should be considered to provide an extra probability of jumping from these nodes. Finally, let $\lambda\in(0,1)$ be a decision parameter, or \em damping \em factor, which determines the likelihood to move through the graph by using a path in  $\mathcal{G}$ or  by randomly jumping instead, according to  the personalization vector. For more information, we refer the reader to \cite{AFPR}.\\

Formally, let $G=G(\lambda, \bold{u},\bold{v})$ be the \em Google matrix\em, with $\lambda\in(0,1)$ defined as
\begin{equation}\label{eq:google-matrix1}
G=\lambda (P_A+\bold{du}^T)+(1-\lambda)\bold{ev}^T\in\mathbb{R}^{N\times N}.
\end{equation} 

A straightforward observation is that the Google matrix $G$ is row-stochastic, i.e., $G\bold{e} = \bold{e}$ where $\bold{e}=(1,\dots,1)^T$. Recall that the personalization vector $\bold{v}\in\mathbb{R}^{N\times 1}$ is such that all its entries are positive (i.e. $\bold{v}>0$) and satisfies $\bold{v}^T\bold{e} = 1$. In a similar way, the vector $\bold{u}\in\mathbb{R}^{N\times1}$ is such that $\bold{u}>0$ and $\bold{u}^T\bold{e} = 1$.\\

Notice that the Google matrix $G$ is positive and thus, by the Perron-Frobenius Theorem (see \cite[Section 8.3]{Meyer}) there is a unique vector  $\pi= \pi(\lambda;\bold{u};\bold{v})>0$ satisfying $\pi^T\textbf{e}=1$  and $G^T\pi=\pi$, or equivalently $\pi^TG=\pi^T$. This vector is the \textit{PageRank vector} $\pi$ of $G$. \\

The existence of dangling nodes will not affect the results presented for the matrix $G$ (see the end of paper \cite{AFPR} for more details). In this sense, for the sake of simplicity, the graph $\mathcal{G}$ will be assumed to have no dangling nodes.\\

From the definition of the Google matrix in equation (\ref{eq:google-matrix1}) and the properties of the PageRank vector $\pi\in\R^{N\times1}$, a straightforward calculation shows that $\pi$ satisfies the equation
$$\pi^T=\lambda \pi ^T P_A+(1-\lambda)\bold{v}^T.$$

Notice that $\pi^T(I_N-\lambda P_A)=(1-\lambda)\mathbf{v}^T$, where $I_N\in\R^{N\times N}$ is the identity square matrix of order $N$. Therefore, the PageRank vector $\pi$ can be explicitly calculated as
\begin{equation}\label{eq:pi-vector}
\pi^T=(1-\lambda)\mathbf{v}^T(I_N-\lambda P_A)^{-1}=\mathbf{v}^TX(\lambda),
\end{equation}
where $X(\lambda)=(1-\lambda)R(\lambda)$ and $R(\lambda)=(I_N-\lambda P_A)^{-1}$ is the resolvent of $P_A$ defined for all suitable damping factor $\lambda\in (0,1)$ (see \cite{Boldi} for more information on this topic).\\

At this point we might ask, what happens if we compute a new PageRank vector using $\pi\in\R^{N\times1}$ as the personalization vector? Is there a limit for the iteration of the PageRank vector $\pi$ under the action of the matrix $X(\lambda)$? Does the limit depend of the initial personalization vector $\bold{v}\in\R^{N\times1}$? In the next section, we present a result concerning the convergence of this iterative process of the PageRank vector $\pi$ under the action of the matrix $X(\lambda)$, specifically for the case when the row-normalization matrix $P_A$ is irreducible. Additionally, we will show that the limit vector of this iteration is, in fact, the left-hand Perron vector associated to irreducible row-normalization matrix $P_A$. \\

In what follows, and unless otherwise specified,  the notation $\bold{0}$ will denote for simplicity the  zero matrix without further reference to the dimension (which may not necessarily be square).

\section{Iteration for a non-negative irreducible probability matrix $P_A$}
\label{Section:irreducible}

In this section, we present a result on the convergence of the iterative process of the PageRank vector $\pi\in\R^{N\times1}$ under the action of the matrix $X(\lambda)=(1-\lambda)(I_N-\lambda P_A)^{-1}$, whenever $P_A$ is the row-normalization of a non-negative and irreducible square matrix $A$ of order $N$. To this aim, let us recall some useful definitions and results.\\

A square matrix $A$ of order $N\geq2$ is said to be \em reducible \em if there exists a permutation matrix $R$ such
that 
\begin{equation*}
R^TAR=\left(\begin{array}{c|c} 
    X & Y \\ 
    \hline
    \bold{0} & Z \end{array}\right),
\end{equation*}
where $X$ and $Z$ are both square matrices (see \cite[Section 4.4]{Meyer}). Otherwise, $A$ is said to be \em irreducible\em . A non-negative square matrix $A$ of order $N$ is \em primitive \em if it is irreducible and has only one non-zero eigenvalue of maximum modulus (\cite[Section 8.5]{Horn-Johnson}). It is worth mentioning that primitivity plays a key role for the convergence to a unique vector when applying the Power Iteration Method to the matrix $X(\lambda)$ (see \cite[Problem 8.5.P16]{Horn-Johnson}).\\ 

A directed graph $\mathcal{G}=(V,E)$ of $N$ nodes is \em strongly connected \em if for any ordered pair $(i,j)\in V\times V$, $1\leq i,j\leq N$, there exists a directed path of edges in $E$ leading from node $i$ to node $j$ (see \cite[Section 1]{Varga}). In  fact, the following result establishes an equivalence between the irreducibility  of a matrix in terms of its strong connectivity. More precisely,

\begin{lem}(\cite[Theorem 1.17]{Varga}){\label{lemma:strongly-connected-irreducible}}
An $N\times N$  matrix $A$ is irreducible if and only if its directed graph $\mathcal{G}$ associated to $A$ is strongly connected.\\
\end{lem}

Notice that  a non-negative matrix $A\neq0$ is irreducible if and only if for every $1\leq i,j\leq N$, there exists a positive integer $k\geq 1$ such that the $(i,j)$-entry of $A^k$ is positive (see \cite[Chapter 2, Theorem 2.1]{BemannPlemmons}).\\

In the context of the Perron-Frobenius theory, if $A\in\R^{N\times N}$ is a non-negative and irreducible matrix, there exists a unique vector $\bold{q}\in\R^{N\times1}$, called the \textit{left-hand Perron vector of the matrix $A$}, such that 
\begin{equation*}
    \bold{q}^TA=r\bold{q}^T\,\qquad\text{with}\qquad \, \bold{q}>0 \qquad\text{and}\qquad \bold{q}^T\bold{e}=1,
\end{equation*}
where $r=\rho(A)$ is the spectral radius of the matrix $A$ (see \cite[Section 8.3]{Meyer}). \\

For $p\geq1$, the $p$-norm of a vector $\bold{x}=(x_1,\dots,x_N)^T\in\R^{N\times 1}$ is defined as $\Vert \bold{x}\Vert_p=\left(\sum_{i=1}^N \vert x_i\vert^p\right)^{1/p}$. Observe that condition $\bold{x}^T\bold{e}=1$ for a non-negative vector $\bold{x}\in\R^{N\times1}$ can be expressed in terms of the $1$-norm as $\Vert \bold{x}\Vert_1=1$. Throughout the paper we mainly use the $p$-norm with $p=1$. \\

The following result asserts that the positivity, irreducibility and row-stochasticity properties of the matrix $P_A$ (the row-normalization of the adjacency matrix $A$) are preserved for the matrix $X(\lambda)=(1-\lambda)(I_N-\lambda P_A)^{-1}$. More precisely,\\

\begin{lem}{\label{lemma:resolvent-row-stochastic}}
Let $A$ be a non-negative irreducible square matrix of order $N$ and let $P_A$ be its row-normalization described in Section \ref{Section:notation}. For a fixed $\lambda\in (0,1)$, the matrix $X(\lambda)=(1-\lambda)(I_N-\lambda P_A)^{-1}$ is positive, irreducible and row-stochastic, i.e. $X(\lambda)\bold{e}=\bold{e}$.
\end{lem}
\begin{proof}
Firstly, since $P_{A}\geq0$ is irreducible and $\rho(P_{A})=1$, the matrix 
$X(\lambda)
=(1-\lambda)(I- \lambda P_{A})^{-1}$ exists and is non-negative for all $\lambda\in(0,1)$ (see \cite[Theorem 3, Section 3, Chapter XIII]{Gantmacher}). In fact, the  matrix $X(\lambda)$ which is given by the Neumann series 
\begin{equation*}
\displaystyle X(\lambda)=(1-\lambda)(I_N-\lambda P_A)^{-1}=(1-\lambda)\sum_{i=0}^{\infty}\left(\lambda P_A\right)^i=(1-\lambda)\left[I_N+(\lambda P_A)+ (\lambda P_A)^2+\dots\right]  
\end{equation*} 
is clearly positive by Lemma \ref{lemma:strongly-connected-irreducible} applied to the irreducible matrix $P_A$. 
If we denote by 
$$B = I_N + (\lambda P_A)^2 + (\lambda P_A)^3+\dots = I_N+\sum_{i=2}^{\infty}\left(\lambda P_A\right)^i,$$
we get $R(\lambda)=(I_N-\lambda P_A)^{-1}=\lambda P_A+ B$. This implies that $R(\lambda)$ is irreducible since  $\lambda P_A$ is so and $B$ is  non-negative with the same order (see \cite{BemannPlemmons, SchwarzA, SchwarzB} for details). Therefore, for a fixed $\lambda\in(0,1)$, we conclude that the matrix $X(\lambda)=(1-\lambda)R(\lambda)$ is a positive and irreducible square matrix of order $N$.\\

Now, since $P_A$ is  row-stochastic we have $P_A\bold{e}=\bold{e}$, which implies that $(P_A)^i\bold{e}=\bold{e}$, for all $i\geq1$. Therefore, since $\lambda\in (0,1)$ we have
\begin{equation*}
    X(\lambda)\bold{e}=(1-\lambda)\sum_{i=0}^{\infty}\left(\lambda P_A\right)^i\bold{e}=(1-\lambda)\sum_{i=0}^{\infty}\lambda^i\,\left(P_A^i\bold{e}\right)=(1-\lambda)\sum_{i=0}^{\infty}\lambda^i\bold{e}=(1-\lambda)\frac{1}{1-\lambda}\bold{e}=\bold{e},
\end{equation*}
and the proof is completed.

\end{proof} 

\begin{rmk}{\label{rmk:resolvent-isometry}}
Notice that the previous result implies that $X(\lambda)$ is a linear $\Vert\cdot\Vert_1$-isometry in the positive cone, that is: if $\bold{x}\in \R^{N\times 1}$ with $\bold{x}\ge 0$ and $\Vert \bold{x}\Vert_1=1$, then $\Vert \bold{x}^TX(\lambda)\Vert_1=1$. This is visualized as the matrix $X(\lambda)$ mapping vectors from the unit circle lying in the positive cone into the unit circle, and  has a nice consequence when the Power Iteration Method is applied to the positive matrix $X(\lambda)$ (see \cite[Problem 8.5.P16 and Theorem 8.2.8]{Horn-Johnson}), in as far as the $\Vert\cdot\Vert_1$-normalization step in this method is no longer needed.\\
\end{rmk}

At this point, we are in position to state the main result of this paper. It shows that the iteration of the PageRank vector $\pi\in\R^{N\times1}$ under the action of the matrix $X(\lambda)$ converges to the left-hand Perron vector of the row-normalization matrix $P_A$ whenever $A$ is a non-negative and irreducible square matrix. Remarkably, the convergence to the left-hand Perron vector of $P_A$ occurs regardless of the personalization vector $\bold{v}\in\R^{N\times1}$ considered. More precisely,\\

\begin{thm}\label{thm:main-theorem}
Let $A$ be a non-negative irreducible square matrix of order $N$ and let $P_A$ be its row-normalization described in Section \ref{Section:notation}. For a fixed $\lambda\in (0,1)$, let consider $X(\lambda)=(1-\lambda)(I_N-\lambda P_A)^{-1}$. Then, for any personalization vector $\bold{v}\in\R^{N\times 1}$ with $\|\bold{v}\|_1=1$, the recursive sequence $\{\bold{x}_k\}_{k\geq0}\subset\R^{N\times1}$, with $\bold{x}_0=\bold{v}$, defined as 
\begin{equation*}
 \bold{x}_k^{T}=\bold{x}_{k-1}^{T}X(\lambda)=\bold{v}^TX(\lambda)^k\,,\quad k\geq 1,
\end{equation*}
converges to the left-hand Perron vector $\bold{c}\in\R^{N\times 1}$ of the matrix $P_A$.
\end{thm}

\begin{proof}
From the assumption trivially follows that  the row-normalization matrix $P_A$ is also irreducible. On the other hand, from the definition of the entries $p_{ij}$ in the non-negative matrix $P_A$, the inequality $\min_i\sum_{j}p_{ij}\leq\rho(P_A)\leq\max_i\sum_{j}p_{ij}$ implies that $\rho(P_A)=1$ (see \cite[Theorem 8.1.22]{Horn-Johnson}). Now, by the Perron-Frobenius Theorem (see \cite[Section 8.3]{Meyer}) applied to the irreducible and non-negative matrix $P_A$, there exists a unique positive vector $\bold{c}\in\mathbb{R}^{N\times 1}$ with $\Vert\bold{c}\Vert_1=1$ (the left-hand Perron vector) associated to the eigenvalue $\rho(P_A)=1$, such that $\bold{c}^TP_A=\bold{c}^T$.\\

Now, by Lemma \ref{lemma:resolvent-row-stochastic} and for a fixed $\lambda\in(0,1)$, the matrix $X(\lambda) = (1-\lambda)(I_N-\lambda P_A)^{-1}$ is positive, irreducible and row-stochastic. Moreover, by the Perron-Frobenius Theorem applied to the matrix $X(\lambda)$, there exists a unique positive vector $\bold{u}\in\mathbb{R}^{N\times 1}$ with $\Vert\bold{u}\Vert_1=1$ such that $\bold{u}^TX(\lambda)=r\bold{u}^T$ where $r$ is the spectral radius of $X(\lambda)$. It is worth mentioning that $r=1$ also follows from the observation that the eigenvectors of $X(\lambda)$ are the same as the eigenvectors of $P_A$, and the corresponding eigenvalues are $\{(1-\lambda)(1-\lambda\beta_i)^{-1}\}$, where $\{\beta_i\}$ are the eigenvalues of $P_A$. A straightforward computation shows that $\vert( 1-\lambda)(1-\lambda\beta_i)^{-1}\vert\leq 1$ whenever $\vert \beta_i\vert\leq1$ and $( 1-\lambda)(1-\lambda\beta_i)^{-1}=1$ if and only if $\beta_i=1$. In fact, the uniqueness of the left-hand Perron vector for the positive matrix $X(\lambda)$ implies that $\bold{u}=\bold{c}$.\\

Finally, for a given vector $\bold{x}_0\in\mathbb{R}^{N\times 1}$, let us consider a recursive sequence of vectors $\{\bold{x}_k\}_{k\geq0}$ in $\R^{N\times1}$ defined by
\begin{equation}\label{eq:sequence -xk-main-theorem}
 \bold{x}_k^{T}=\bold{x}_{k-1}^{T}X(\lambda),\quad (k\geq 1)\,,\,\qquad\text{with}\qquad\Vert \bold{x_0}\Vert_1=1.
\end{equation}
A straightforward computation shows that $\bold{x}_k^{T}=\bold{v}^TX(\lambda)^k$, $k\geq 1$, with $\bold{x}_0=\bold{v}>0$. Since $X(\lambda)$ is a positive matrix, by the Power Iteration Method (see \cite[Problem 8.5.P16 and Theorem 8.2.8]{Horn-Johnson}) and Remark \ref{rmk:resolvent-isometry} applied to the matrix $X(\lambda)$ we conclude that the sequence $\{\bold{x}_k\}_{k\geq0}$ in equation (\ref{eq:sequence -xk-main-theorem}) converges to a non-zero vector $\bold{x}\in\R^{N\times 1}$ with $\Vert\bold{x}\Vert_1=1$ such that
\begin{equation*}
    \bold{x}^T=\lim_{k\to\infty}\bold{x}_k^T
    =\lim_{k\to\infty}\bold{v}^TX^{k}(\lambda)
    =\bold{c}^T,
\end{equation*}
where $\bold{c}\in\R^{N\times 1}$ is the left-hand Perron vector of the matrix $P_A$ associated to the eigenvalue $\rho(P_A)=1$.
\end{proof}

In Figure \ref{fig:irreducible-case} we show a numerical simulation of Theorem \ref{thm:main-theorem}. Let us consider a directed graph $\mathcal{G}=(V,E)$ where $V=\{1,2,\dots,10\}$ is the set of nodes and let $A=(a_{ij})_{i,j}$ be its adjacency square matrix of order $10$. A straightforward observation is that the directed graph $\mathcal{G}$ in Figure \ref{fig:irreducible-case}(a) is strongly connected and therefore, by Lemma \ref{lemma:strongly-connected-irreducible}, the row-normalization matrix $P_A$ of $A$ is a non-negative, irreducible and row-stochastic square matrix of order $10$.\\

Now, for a fixed $\lambda\in (0,1)$, we compute the recursive sequence of vectors $\{\bold{x}_k\}_{k\geq0}$ in $\R^{10\times1}$ defined as
\begin{equation}\label{eq:different-personalization}
 \bold{x}_k^{T}=\bold{x}_{k-1}^{T}X(\lambda)=\bold{v}^TX(\lambda)^k\,,\quad (k\geq 1),
\end{equation}
using three different personalization vectors: 
\begin{itemize}
    \item A balanced vector $\bold{v}\in \R^{10\times1}$ with $\Vert \bold{v}\Vert_1=1$, where its $1$-norm is the uniformly distributed among the $10$ nodes. This vector corresponds to the color blue in Figure \ref{fig:irreducible-case}(b).
    \item A first unbalanced vector $\bold{v}_{\downarrow}\in \R^{10\times1}$ with $\Vert \bold{v}_{\downarrow}\Vert_1=1$, where the 1\% of its $1$-norm is assigned to node $1$ and the remaining 99\% is uniformly distributed among nodes $2$ through $10$. This vector is represented in orange in Figure \ref{fig:irreducible-case}(b).
    \item A second unbalanced vector $\bold{v}_{\uparrow}\in \R^{10\times1}$ with $\Vert \bold{v}_{\uparrow}\Vert_1=1$, where the 91\% of its $1$-norm is assigned to node $1$ and the remaining 9\% is uniformly distributed among nodes $2$ through $10$. This vector is represented in green in Figure \ref{fig:irreducible-case}(b).
\end{itemize}
Since node 1 takes a different value from the remaining nodes in the unbalanced cases, we highlight this node by coloring it purple, as it can be observed in Figure \ref{fig:irreducible-case}(a).\\

For the directed graph in Figure \ref{fig:irreducible-case}(a), the left-hand Perron vector $\bold{c}\in\R^{10\times1}$ corresponding to $P_A$ can be explicitly calculated using symbolic computation software. Specifically, $\bold{c}=\frac{1}{101}(13,14,14,14,20,12,6,4,2,2)^T$ with $\bold{c}>0$ and $\Vert \bold{c}\Vert_1=1$. With this left-hand Perron vector in hand, Figure \ref{fig:irreducible-case}(b) illustrates how the $1$-norm distance between the vector $\bold{c}\in\R^{10\times1}$ and the (first seven) vectors of sequence $\{\bold{x}_k\}_{k\geq1}$ in $\R^{10\times1}$ defined by equation (\ref{eq:different-personalization}) tends to zero for the different personalization vectors $\bold{v}, \bold{v}_{\downarrow}$, and $\bold{v}_{\uparrow}$ described above.\\
 
\begin{figure}[h]
        \centering
        \includegraphics[width=1\linewidth]{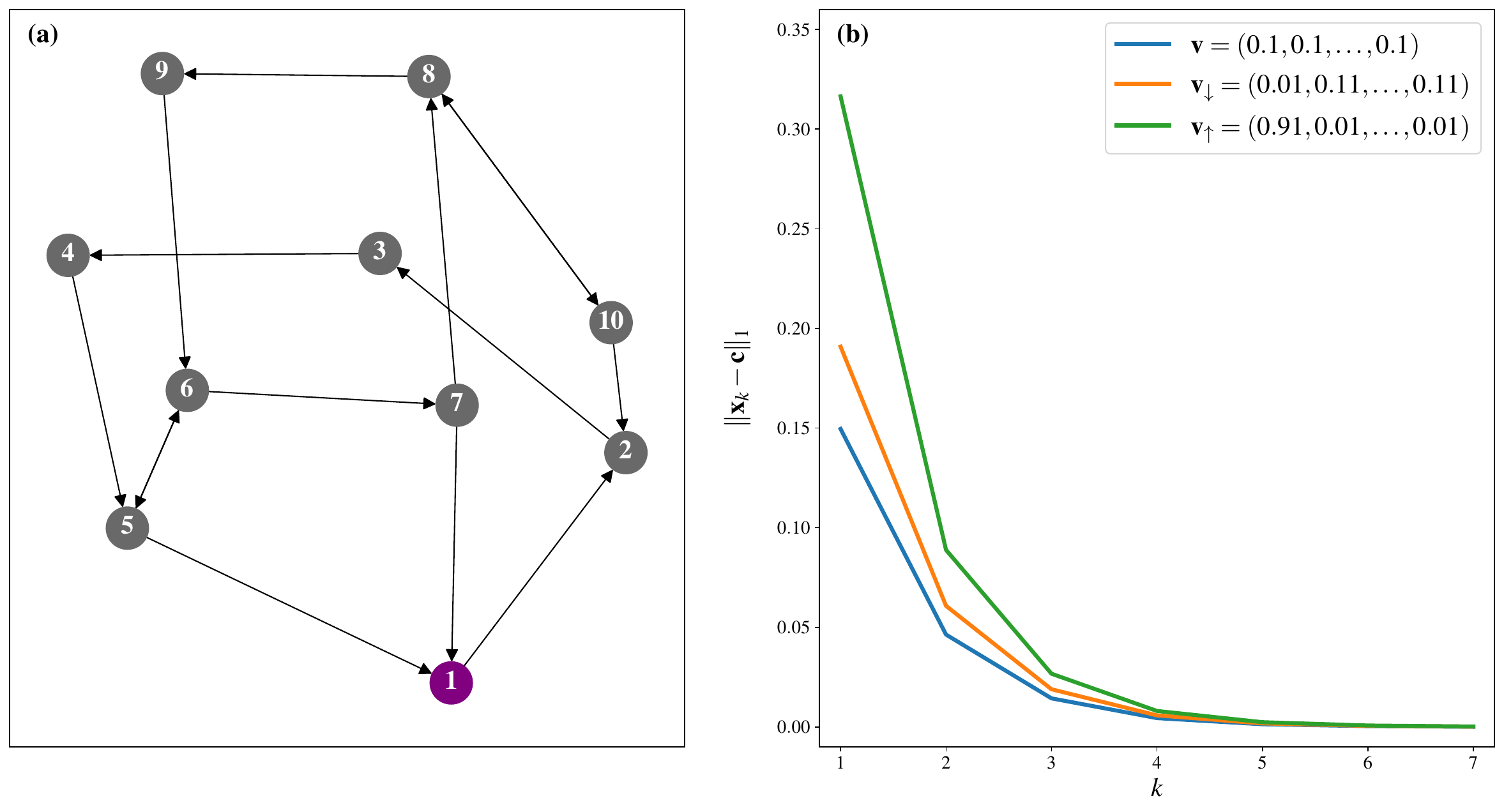}
        \caption{Numerical simulation for the irreducible case of the matrix $P_A$ in Theorem \ref{thm:main-theorem}.}
        \label{fig:irreducible-case}
    \end{figure}
In the next section, we will investigate the convergence of the PageRank vector $\pi$ under the iteration of the matrix $X(\lambda) = (1-\lambda)(I_N-\lambda P_A)^{-1}$ when the row-normalization matrix $P_A$ is no longer irreducible. To this aim, we will look at some configurations for the matrix $P_A$.\\

Note that it is straightforward to check that Theorem~\ref{thm:main-theorem} solves the existence and uniqueness of fixed points of the PageRank as a function $PR_\lambda: \Delta_N^+\longrightarrow \Delta_N^+$ for every damping factor $\lambda\in (0,1)$ if the network considered is strongly connected. Therefore, it can be rewritten as follows,\\

\begin{coro}\label{cor:FixedPointStronglyConnected}
If $\mathcal{G}=(V,E)$ be a strongly connected directed graph, then for every damping factor $\lambda\in (0,1)$, there is a unique fixed point of the PageRank as a function $PR_\lambda: \Delta_N^+\longrightarrow \Delta_N^+$. Furthermore, such fixed point is the left-hand Perron vector  of the matrix $P_A$.
\end{coro}
 
Since Corollary~\ref{cor:FixedPointStronglyConnected} states that the only fixed point of the PageRank is  the left-hand Perron vector  of the matrix $P_A$ in the case of strongly connected directed graphs, the computation of such fixed point can be computationally simplified in some cases, as the following result shows.\\

\begin{thm}\label{thm:PageRank-Degree}
Let $\mathcal{G}=(V,E)$ be a graph where $V=\{1,\dots,N\}$ is the set of nodes.
\begin{itemize}
 \item[{\it (i)}] If $\mathcal{G}$ is an undirected and connected graph, then the left-hand Perron vector of the matrix $P_A$ is 
 \[
 \bold{c}=\frac 1{2|E|}\left(k(1),\cdots,k(N)\right)^T,
 \]
 where $k(i)$ denotes the degree of node $i$, for every $1\le i\le N$.
 \item[{\it (ii)}] If $\mathcal{G}$ is a directed and strongly connected graph, then then the left-hand Perron vector  of the matrix $P_A$ is 
 \[
 \bold{c}=\frac 1{|E|}\left(k_{out}(1),\cdots,k_{out}(N)\right)^T
 \]
if and only if $\mathcal{G}$ admits some Eulerian path (i.e. there is a closed path in $\mathcal{G}$ that visits every edge exactly once). 
\end{itemize} 
\end{thm}

\begin{proof}
Before completing the proof, the following fact must be checked:
\begin{itemize}
    \item [\textbf{Fact 1.}] If $A,\,D$ are two square matrices and $D$ is diagonal, then $\bold{w}\in\R^{N\times1}$ is a left eigenvector of $D^{-1}A$ associated to the eigenvalue $\mu$ if and only if $D^{-1}\bold{w}$ is a  left eigenvector of $AD^{-1}$ associated to the eigenvalue $\mu$, where $D$ is any diagonal matrix.
 \end{itemize}
 In order to prove this fact notice, on the one hand, that if $\bold{w}\in\R^{N\times1}$ is a left eigenvector of $D^{-1}A$ associated to the eigenvalue $\mu$, then
 \[
 (D^{-1}\bold{w})^T AD^{-1}=\bold{w}^T (D^{-1})^T AD^{-1}=\left(\bold{w}^TD^{-1}A\right)D^{-1}=\mu \bold{w}^TD^{-1}=\mu\left(D^{-1}\bold{w}\right)^T,
 \]
 since $D$ is a diagonal matrix. Hence, $D^{-1}\bold{w}$ is a  left eigenvector of $AD^{-1}$ associated to the eigenvalue $\mu$. 
 
On the other hand, if $D^{-1}\bold{w}$ is a  left eigenvector of $AD^{-1}$ associated to the eigenvalue $\mu$, then 
\[
\bold{w}^T\left(D^{-1}A\right)=\bold{w}^TD^{-1}AD^{-1}D=\left(D^{-1}\bold{w}\right)^T AD^{-1}D=\mu \left(D^{-1}\bold{w}\right)^TD=\mu \bold{w}^T\left(D^{-1}\right)^TD=\mu\bold{w}^T,
\]
so $\bold{w}\in\R^{N\times1}$ is a left eigenvector of $D^{-1}A$ associated to the eigenvalue $\mu$. 

\medskip Once  {\bf Fact 1} has been proved, we can start proving {\it (i)}. Since $P_A=D^{-1}A$ is row-stochastic matrix and $\mathcal{G}$ is connected, it follows that $\rho(P_A)=\rho(P_A^T)=1$ and the corresponding (left) eigenspaces associated to eigenvalue $\mu=1$ are 1-dimensional. In addition to this, since $P_A^T=(D^{-1}A)^T=AD^{-1}$ is column-stochastic, then $\bold{e}=(1,\cdots,1)^T$ is a left eigenvector of $P_A^T$ associated to eigenvalue $\mu=1$. Hence, by {\bf Fact~1}, $D\bold{e}$ is a left eigenvector of $P_A$ associated to eigenvalue $\mu=1$ and therefore
\[
\frac 1{\displaystyle\sum_i k(i)}D\bold{e}=\frac 1{2|E|}\left(k(1),\cdots,k(N)\right)^T
\]   
is the left-hand Perron vector of $P_A$.

\medskip In order to proof {\it (ii)}, note that  since $P_A=D^{-1}A$, then
\[
AD^{-1}=
\left(
\begin{array}{ccc}
\frac{a_{11}}{k_{out}(1)}& \cdots& \frac{a_{1N}}{k_{out}(N)}\\
\cdots & \cdots & \cdots\\
\frac{a_{N1}}{k_{out}(1)} & \cdots & \frac{a_{NN}}{k_{out}(N)}
\end{array}
\right).
\] 
Hence the sum of each column of  $AD^{-1}$ is $k_{in}(1)/k_{out}(1),\,\cdots,\, k_{in}(N)/k_{out}(N)$ respectively, where $k_{in}(j)=\sum_{i}a_{ij}$ for every $1\le j\le N$. Therefore, by using {\bf Fact\,1},
\[
 \bold{c}=\frac 1{|E|}\left(k_{out}(1),\cdots,k_{out}(N)\right)^T
\] 
is the normalized left-hand Perron vector of $D^{-1}A$ if and only if $\bold{e}$ is a left eigenvector of $AD^{-1}$ associated to eigenvalue $\mu=1$, but 
\[
\bold{e}^T AD^{-1}= \left(\frac {k_{in}(1)}{k_{out}(1)},\,\cdots,\, \frac{k_{in}(N)}{k_{out}(N)}\right),
\] 
which makes that $\bold{e}$ is the normalized left-hand Perron vector of $AD^{-1}$ if and only $k_{in}(j)=k_{out}(j)$ for every $1\le j\le N$, but this is equivalent to the fact that $\mathcal{G}$ admits some Eulerian paths (see, for example \cite{Diestel}).
\end{proof}


\section{Iteration for a non-negative reducible probability matrix $P_A$}
\label{Section:reducible}

In this section, we investigate the convergence of the Power Iteration Method applied to the matrix $X(\lambda)=(1-\lambda)(I_N-\lambda P_A)^{-1}$ for the case when the row-normalization matrix $P_A$ is no longer irreducible. To this aim, we will assume (see \cite[Section 2.3]{Varga}) that for the reducible square adjacency matrix $A=(a_{ij})_{i,j}$ of order $N$ defined in Section \ref{Section:notation} there exists a permutation matrix $S$ such that the matrix $SAS^T$ has an upper-triangular form
\begin{equation*}
SAS^T=\left(\begin{array}{c|c} 
    A_1 & B \\ 
    \hline
    \bold{0} & A_2 \end{array}\right),
\end{equation*}
where $A_1$ is either an irreducible or identically zero square matrix of order $n_1$ and $A_2$ is an irreducible square matrix of order $n_2$, with $n_1+n_2=N$, and $B$ is a $n_1\times n_2$ matrix. Notice that the absence of dangling nodes (see Section \ref{Section:notation}) means that $A_2$ cannot be  the zero matrix. At this point, we distinguish  three different situations for the  matrix $SAS^T$: 
\begin{itemize}
\addtolength{\itemindent}{3.5cm}
    \item[\textbf{Diagonal case:}] Where $B=\bold{0}_{n_1\times n_2}$, but $A_1$ and $A_2$ are non-zero square matrices (Theorem \ref{thm:diagonal-case}).
    \item[\textbf{Zero-block column case:}]  Where $A_1=\bold{0}_{n_1\times n_1}$, but $B$ and  $A_2$ are non-zero matrices (Theorem \ref{thm:zero-case}).
    \item[\textbf{General reducible case:}]  Where $A_1, B$, and $A_2$ are non-zero matrices (Theorem \ref{thm:non-zero-case}).
\end{itemize}

For all three cases we will prove not only the convergence of the Power Iteration Method applied to the matrix $X(\lambda)$, but also the form of the vector to which the recursive sequence $\{\bold{x}_k\}_{k\geq0}$ in $\R^{N\times1}$ defined as 
\begin{equation}\label{eq:recursive-sequence}
 \bold{x}_k^{T}=\bold{x}_{k-1}^{T}X(\lambda)=\bold{v}^{T}X^k(\lambda)\,,\quad (k\geq 1),
\end{equation}
converges. Notice that  the row-stochasticity of the matrix $P_A$ together with the damping factor being strictly smaller than 1, imply the existence of the non-negative matrix $X(\lambda)=(I_N-\lambda P_A)^{-1}$, but it can be no longer assumed irreducible. \\

For instance, in the \textbf{Diagonal case}, we will show that the sequence in equation (\ref{eq:recursive-sequence}) converges to a non-zero vector of the form $(\enspace\alpha_1\bold{c}_1^T\quad | \quad \alpha_2\bold{c}_2^T\enspace)$, for some $\alpha_1,\alpha_2\geq0$ and $\alpha_1+\alpha_2=1$, where $\bold{c}_1\in \R^{n_1\times 1}$ and $\bold{c}_2\in \R^{n_2\times 1}$ are the left-hand Perron vectors for the row-normalization matrices $P_{A_1}$ and $P_{A_2}$, respectively. Roughly speaking, in the diagonal case where clusters $A_1$ and $A_2$ are disjoint, the $1$-norm of the PageRank vector $\pi\in \R^{N\times1}$ is distributed (perhaps not equitably) between vectors $\bold{c}_1\in\R^{n_1\times 1}$ and $\bold{c}_2\in\R^{n_2\times 1}$.\\

On the other hand, for the \textbf{Zero-block column case} and the \textbf{General reducible case} where nodes in cluster $A_1$ are connected to nodes in cluster $A_2$ (but not conversely), we will show that the sequence in equation (\ref{eq:recursive-sequence}) converges to non-zero vector of the form $(\enspace\bold{0}_{n_1\times1}^T\quad | \quad \bold{c}_2^T\enspace)$, where $\bold{c}_2\in \R^{n_2\times 1}$ is the left-hand Perron vector of matrix $P_{A_2}$. This fact can be understood as follows: The $1$-norm of the PageRank vector $\pi\in \R^{N\times1}$ is divided between clusters $A_1$ and $A_2$. However, over time, the portion of this $1$-norm in cluster $A_1$ is transferred to cluster $A_2$ through the connections represented by matrix $B$.  Ultimately, all the $1$-norm of vector $\pi$ asymptotically ends up in cluster $A_2$ and finally, by Theorem \ref{thm:main-theorem} applied to matrix $P_{A_2}$, the sequence $\{\bold{x}_k\}_{k\geq0}$ defined in equation (\ref{eq:recursive-sequence}) converges to the vector $(\enspace\bold{0}_{n_1\times1}^T\quad | \quad \bold{c}_2^T\enspace)$. In this sense, the cluster $A_2$ is what we call a \textit{Dangling Cluster}, that is, a set of nodes that has no outgoing links to other clusters within the network.\\

The situation described above is illustrated in Figure \ref{fig:dangling-cluster}. For this numerical simulation, we consider a graph with $N=20$ nodes which is in fact two copies of the graph in Figure \ref{fig:irreducible-case}(a). One of the copies plays the role of Cluster  and the second copy plays the role of Dangling Cluster by simply adding one edge that goes from the first copy to the second copy, which is represented by the black arrow in Figure \ref{fig:dangling-cluster}(a).
Notice that the term Dangling Cluster, as it was mentioned above, refers to a group of nodes with no out-links to any other cluster in the network, and it can be considered an extension of the notion of a dangling node to a group of nodes in the network. Additionally, in Figure \ref{fig:dangling-cluster}(b), we compute a heat map where we represent the different normalized values of nodes (in rows) $\{w_1,\dots,w_{10}\}$ in the Cluster and $\{u_1,\dots,u_{10}\}$ in the Dangling Cluster  along the $k$-th iteration (in columns) under the action of the matrix $X(\lambda)$.
\begin{figure}[h!]
        \centering
        \includegraphics[width=1\linewidth]{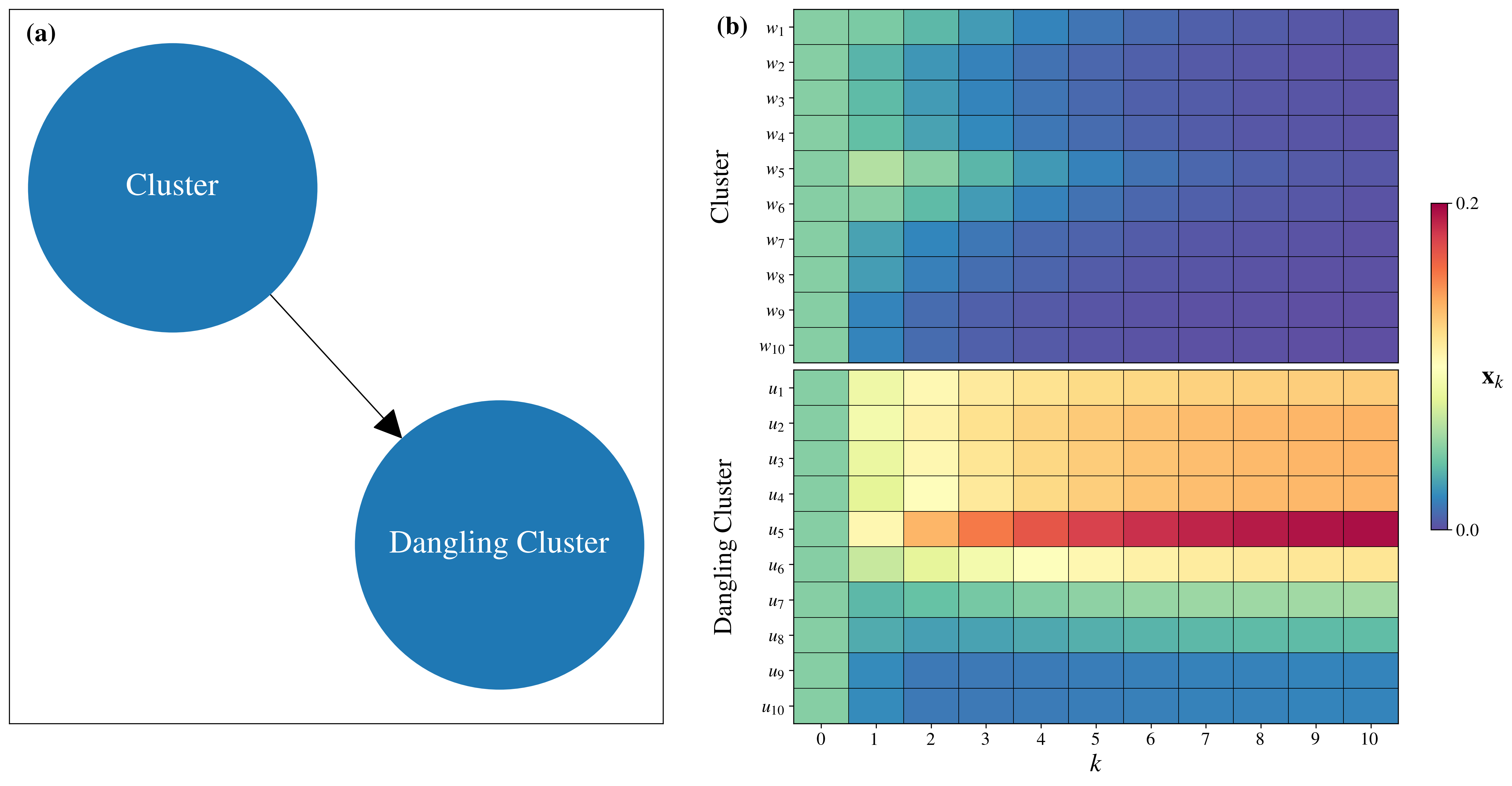}
        \caption{A numerical illustration of the converge for the Zero-block column and the General reducible cases.}
        \label{fig:dangling-cluster}
    \end{figure}

Initially, the personalization vector $\bold{v}\in\R^{20\times1}$ is considered to be uniform, that is,  $\bold{v}=\displaystyle\frac{1}{20}\bold{e}=(0.05\,,\dots\,,\,0.05)^T$ with $\bold{e} = (1,\dots,1)^T\in\R^{20\times1}$, as can be observed in the first column ($k=0$), formed by green rectangles, in Figure \ref{fig:dangling-cluster}(b). As the iteration $\bold{v}^TX^k(\lambda)$ is calculated, the color of each node in the directed graph begins to change, with blue color representing values close to $0$ and red color indicating values close to $0.2$. As it can be observed in Figure \ref{fig:dangling-cluster}(b), the redistribution of the $1$-norm of the uniform personalization vector $\bold{v}$ is such that the vector $\bold{v}^TX^k(\lambda)$ with $k=10$ shows a significant shift in the values, which are very close to zero for all nodes in the Cluster part of the network. In this sense, the entire initial $1$-norm of the personalization vector ends up being redistributed among the nodes in the Dangling Cluster part of the network. Notice that  the graph in the Dangling Cluster part is the same as in Figure \ref{fig:irreducible-case}(a) and  the left-hand Perron vector $\bold{c}\in\R^{10\times1}$ associated with the row-normalization matrix $P_A$ is $\bold{c}=\frac{1}{101}(13,14,14,14,20,12,6,4,2,2)^T$. The fact that the  fifth entry of  $\bold{c}$  exhibits the highest  value is consistent with the red color in Figure \ref{fig:dangling-cluster}(b), which corresponds to node $u_5$ in the Dangling Cluster aftert the $10^{th}$ iteration ($k=10$).\\

Now, once we have an idea of how the iteration of the matrix $X(\lambda)=(1-\lambda)(I_N-\lambda P_A)^{-1}$ behaves when the adjacency matrix $A$ is reducible (and so is  its row-normalization $P_A$), we are in a position to state and prove the results concerning the different cases mentioned above, starting with the \textbf{Diagonal case}. In what follows, for simplicity of notation, we will refer to the reduced matrix $SAS^T$ simply as $A$.\\

\begin{thm}\label{thm:diagonal-case}
Let $A$ be a non-negative reducible square matrix of order $N$ of the form
\begin{equation*}
    A=\left(\begin{array}{c|c} 
    A_1 & \bold{0} \\ 
    \hline
    \bold{0} & A_2 \end{array}\right),
\end{equation*}
where $A_1$ and $A_2$ are irreducible square matrices of order $n_1$ and $n_2$, respectively, with $n_1+n_2=N$. Let $P_{A_1}$ and $P_{A_2}$ be the row-normalization of matrices $A_1$ and $A_2$, respectively. 
For a fixed $\lambda\in (0,1)$ consider $X(\lambda)=(1-\lambda)(I_N-\lambda P_A)^{-1}$. Then for any personalization vector $\bold{v}\in\R^{N\times 1}$ with $\|\bold{v}\|_1=1$, the recursive sequence $\{\bold{x}_k\}_{k\geq0}\subset\R^{N\times1}$, with $\bold{x}_0=\bold{v}$, defined as 
\begin{equation*}
 \bold{x}_k^{T}=\bold{x}_{k-1}^{T}X(\lambda)=\bold{v}^TX(\lambda)^k\,,\quad (k\geq 1),
\end{equation*}
converges to the vector $\bold{x}^T=(\enspace\alpha_1\bold{c}_1^T\quad | \quad \alpha_2\bold{c}_2^T\enspace)$, for some $\alpha_1,\alpha_2\geq0$ with $\alpha_1+\alpha_2=1$, where $\bold{c}_1\in\R^{n_1\times 1}$ and $\bold{c}_2\in\R^{n_2\times 1}$ are the left-hand Perron vector for matrices $P_{A_1}$ and $P_{A_2}$, respectively.
\end{thm}

\begin{proof}
Let $P_A$ be the row-stochastic reducible matrix obtained from $A$ by row-normalization.
$$P_A=\left(\begin{array}{c|c} 
P_{A_1} & \bold{0} \\
\hline
\bold{0} & P_{A_2} \end{array}\right).$$
As it was mentioned in Section \ref{Section:notation}, the Google matrix $G$ without dangling nodes given by
\begin{equation*}
G=\lambda P_A+(1-\lambda)\bold{ev}^T\in\mathbb{R}^{N\times N}
\end{equation*} 
is a row-stochastic square matrix for which there exists a unique positive vector $\pi\in\R^{N\times1}$ (called PageRank vector), $\pi^T\bold{e}=1$, satisfying the equation
\begin{equation}\label{eq:pagerank-diagonal-case}
    \pi^T=(1-\lambda)\mathbf{v}^T(I_N-\lambda P_A)^{-1}=\mathbf{v}^TX(\lambda),
\end{equation}
where $X(\lambda)=(1-\lambda)(I_N-\lambda P_A)^{-1}$. A straightforward computation shows that the inverse of the reducible square matrix $I_N-\lambda P_A$ is the reducible matrix given by
\begin{align*}
    (I_N-\lambda P_A)^{-1}=\left(\begin{array}{c|c} 
    (I_{n_1}- \lambda P_{A_1})^{-1} & \bold{0}\\ 
    \hline
    \bold{0} & (I_{n_2}-\lambda P_{A_2})^{-1} 
    \end{array}\right).
\end{align*}
Therefore, in this diagonal case, equation (\ref{eq:pagerank-diagonal-case}) can be rewritten as follows
\begin{align*}
    \pi^T
    =\bold{v}^T X(\lambda)
    =\bold{v}^T(1-\lambda)(I_N-\lambda P_A)^{-1}
    &=\bold{v}^T(1-\lambda)\left(\begin{array}{c|c} 
    (I_{n_1}- \lambda P_{A_1})^{-1} & \bold{0} \\ 
    \hline
    \bold{0} & (I_{n_2}-\lambda P_{A_2})^{-1} 
    \end{array}\right)\\
    &=\bold{v}^T\left(\begin{array}{c|c} 
    (1-\lambda)(I_{n_1}- \lambda P_{A_1})^{-1} & \bold{0} \\
    \hline
    \bold{0} & (1-\lambda)(I_{n_2}-\lambda P_{A_2})^{-1} 
    \end{array}\right)\\
    &=\bold{v}^T\left(\begin{array}{c|c} 
    X_1(\lambda) & \bold{0} \\ 
    \hline
    \bold{0} & X_2(\lambda)
    \end{array}\right),
\end{align*}
where $X_1(\lambda)=(1-\lambda)(I_{n_1}-\lambda P_{A_1})^{-1}$ and $X_2(\lambda)=(1-\lambda)(I_{n_2}-\lambda P_{A_2})^{-1}$. Now, let us consider a recursive sequence of vectors $\{\bold{x}_k^T\}_{k\geq0}$ in $\R^{N\times1}$ defined by
\begin{equation}\label{eq:sequence-xk-diagonal-case}
 \bold{x}_k^{T}=\bold{x}_{k-1}^{T}X(\lambda)=\bold{v}^TX(\lambda)^k\,,\quad (k\geq 1)\,,\,\qquad\text{with}\qquad \, \bold{x}_0=\bold{v}>0 \qquad\text{and}\qquad\Vert \bold{v}\Vert_1=1.
\end{equation}
Let $\bold{v}^T=(\enspace\bold{v}_1^T\quad|\quad\bold{v}_2^T\enspace)$ be a decomposition of the personalization vector $\bold{v}\in\R^{N\times1}$ where $\bold{v}_1\in\R^{n_1\times1}$ and $\bold{v}_2\in\R^{n_2\times1}$. Clearly  
\begin{equation*}
    \bold{x}_k^T
    =\bold{v}^TX^{k}(\lambda)=(\enspace \bold{v}_1^T\quad|\quad\bold{v}_2^T\enspace)
    \left(\begin{array}{c|c}
    X_1^k(\lambda) & \bold{0} \\
    \hline
    \bold{0} & X_2^k(\lambda)
    \end{array}\right)
    =    \left(\enspace\bold{v}_1^TX_1^k(\lambda)\quad\Big{|}\quad \bold{v}_1^T X_2^k(\lambda)\enspace \right).
\end{equation*}
Now, since $A_1$ and $A_2$ are non-negative irreducible square matrices, by Theorem \ref{thm:main-theorem} in Section \ref{Section:irreducible} we have
\begin{equation}\label{alphasdiagonal}
    \lim_{k\to\infty}\bold{v}_1^TX_1^k(\lambda)
    =\Vert \bold{v}_1^T\Vert_1\lim_{k\to\infty}\frac{\bold{v}_1^T}{\Vert \bold{v}_1^T\Vert_1}X_1^k(\lambda)=\Vert \bold{v}_1^T\Vert_1 \bold{c}_1^T \qquad,\qquad  \lim_{k\to\infty}\bold{v}_2^TX_2^k(\lambda)=
    \Vert \bold{v}_2^T\Vert_1 \lim_{k\to\infty}\frac{\bold{v}_2^T}{\Vert \bold{v}_2^T\Vert_1}X_2^k(\lambda)=\Vert \bold{v}_2^T\Vert_1 \bold{c}_2^T,
\end{equation}
where $\bold{c}_1\in\R^{n_1\times1}$ and $\bold{c}_2\in\R^{n_2\times1}$ are the left-hand Perron vectors for $P_{A_1}$ and $P_{A_2}$, respectively. Finally, the sequence $\{\bold{x}_k\}_{k\geq0}$ defined in equation (\ref{eq:sequence-xk-diagonal-case}) converges to a non-zero vector $\bold{x}\in\R^{N\times 1}$ with $\Vert \bold{x}\Vert_1=1$ such that
\begin{equation*}
    \bold{x}^T=\lim_{k\to\infty}\bold{x}_k^T
    =\lim_{k\to\infty}\bold{v}^TX^{k}(\lambda)
    =\lim_{k\to\infty}\left(\enspace\bold{v}_1^TX_1^k(\lambda)\quad\Big{|}\quad \bold{v}_1^T X_2^k(\lambda)\enspace\right) 
    =(\enspace\Vert \bold{v}_1^T\Vert_1\bold{c}_1^T\quad \Big{|} \quad \Vert \bold{v}_2^T\Vert_1\bold{c}_2^T\enspace).
\end{equation*}
In fact, since $\Vert\bold{c}_1\Vert_1=1$, $\Vert\bold{c}_2\Vert_1=1$ and $\Vert \bold{v}_1^T\Vert_1+\Vert \bold{v}_2^T\Vert_1=\Vert \bold{v}^T\Vert_1=1$, we conclude the vector $\bold{x}^T$ is a linear convex combination of the vectors $\left(\enspace\bold{c}_1^T\quad \Big{|}\quad \bold{0}_{n_2\times1}^T\enspace\right)$ and $\left(\enspace\bold{0}_{n_1\times1}^T\quad \Big{|} \quad \bold{c}_2^T\enspace\right)$.

\end{proof}

As with Theorem~\ref{thm:main-theorem} we can reinterpret  Theorem~\ref{thm:diagonal-case} to prove  the existence  of fixed points for the PageRank as a function $PR_\lambda: \Delta_N^+\longrightarrow \Delta_N^+$ for every damping factor $\lambda\in (0,1)$,  when the network considered has two completely disconnected strongly connected components. More precisely,\\

 \begin{coro}\label{cor:diagonal-case}
Let $\mathcal{G}=(V,E)$ be a directed graph with two completely disconnected strongly connected components $\mathcal{G}_1$ and $\mathcal{G}_2$, then for every damping factor $\lambda\in (0,1)$, there are infinite many fixed points of the PageRank as a function $PR_\lambda: \Delta_N^+\longrightarrow \Delta_N^+$. Furthermore, the set of such fixed points is the convex hull  in $\Delta_N^+$ of the left-hand Perron vectors of the matrices $P_{A_1}$, $P_{A_2}$ given by $\mathcal{G}_1$ and $\mathcal{G}_2$.\\
\end{coro}

In Figure \ref{fig:diagonal-case} we present a numerical simulation of Theorem \ref{thm:diagonal-case} which illustrates how the limit vector $\bold{x}$ depends on the initial $1$-norm distribution of the personalization vector in each cluster $A_1$ and $A_2$. For this, take $\mathcal{G}=(V,E)$  a directed ten-node  graph whose set of nodes $V=\{1,2,\dots,10\}$  is  divided into two disjoint groups: Cluster $A_1$ contains nodes $1$ through $5$, and cluster $A_2$ contains nodes $6$ through $10$ (see Figure \ref{fig:diagonal-case}(a)). If $P_{A_1}$ and $P_{A_2}$ denote the row-normalization of the adjacency matrix of cluster $A_1$ and $A_2$, respectively, a straightforward computation shows that $\bold{c_1}=\frac{1}{19}(1,8,4,2,4)^T$ is the left-hand Perron vector of $P_{A_1}$ and $\bold{c_2}=\frac{1}{29}(6,6,8,3,6)^T$ is the left-hand Perron vector of $P_{A_2}$.\\

The graph in Figure \ref{fig:diagonal-case}(a)  clearly is not strongly connected. In Figure \ref{fig:diagonal-case}(b) we consider two different cases: Example 1, where the 90\% of the $1$-norm of the personalization vector $\bold{v}\in\R^{10\times1}$ is distributed among nodes in cluster $A_1$ (bottom graph) and the remaining 10\% is distributed among nodes in cluster $A_2$ (top graph), and Example 2, where the distribution is exactly the opposite. Roughly speaking, both in Example 1 and Example 2, the part in each cluster of the limit vector $\bold{x}^T=\displaystyle\lim_{k\to\infty}\bold{x}_k^T$ of the recursive sequence $\{\bold{x}_k\}_{k\geq0}\in\R^{N\times1}$ defined in equation (\ref{eq:sequence-xk-diagonal-case}) is the same, up to a constant proportional to the $1$-norm distribution of the personalization vector among clusters $A_1$ and $A_2$. For instance, in Example 1, the second component of the limit vector $\bold{x}\in\R^{10\times1}$ belonging to cluster $A_1$ (bottom graph), is $0.9$ times the second component of the left-hand Perron vector $\bold{c_1}$ (that is, $0.9*(8/19)\approx0.379$), whereas in Example 2, the same component is only $0.1$ times this component (that is, $0.1*(8/19)\approx0.042$).

\begin{figure}[H]
        \centering
        \includegraphics[width=1\linewidth]{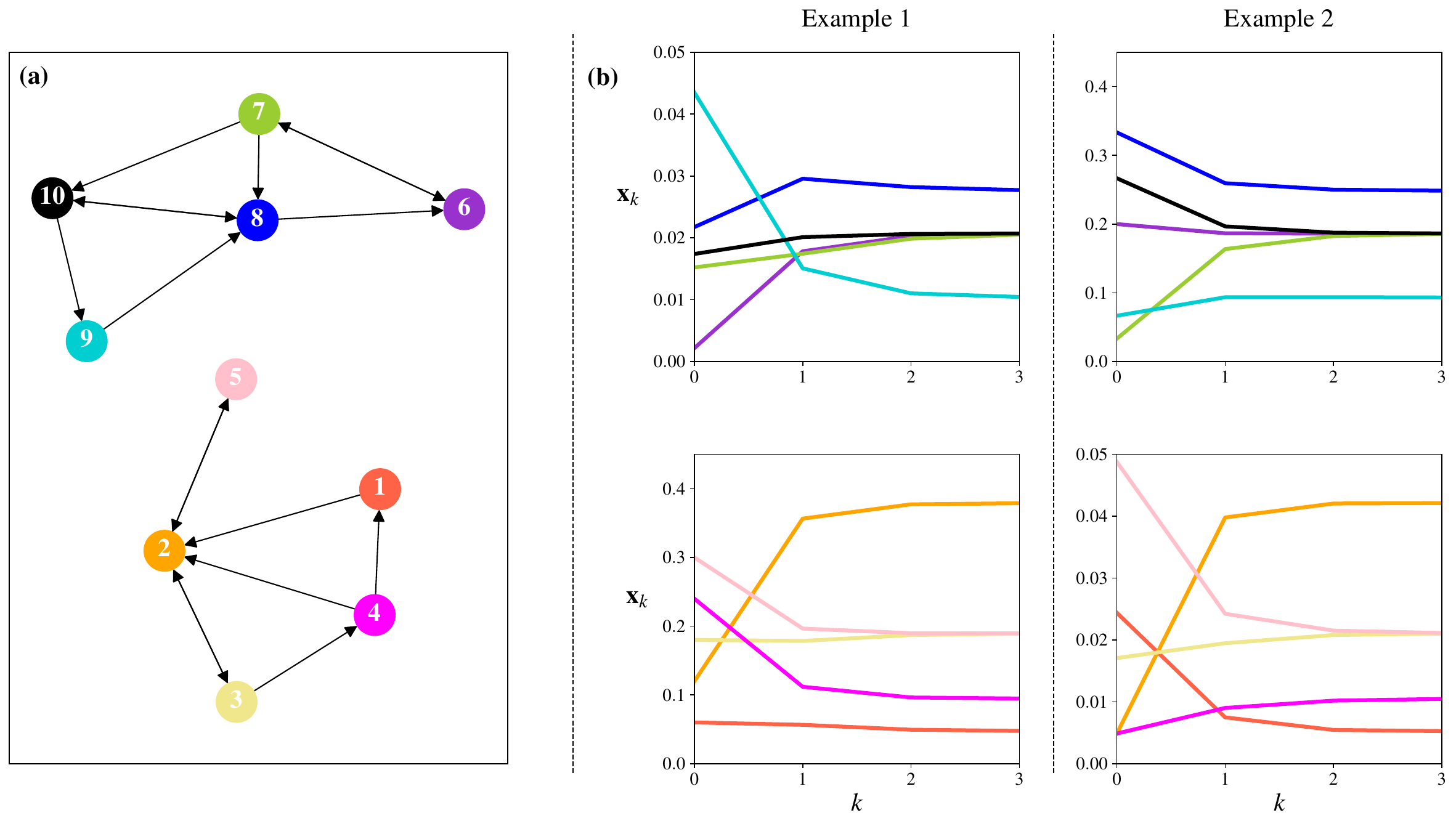}
        \caption{A numerical illustration of the converge for the Diagonal case in Theorem \ref{thm:diagonal-case}.}
        \label{fig:diagonal-case}
    \end{figure}
    
In the sequel, we will make use of the following remark concerning the matrix norm.\\

\begin{rmk}\label{rmk:operator-norm}
Given a square matrix $Q$ of order $N$, the matrix norm $\|Q\|_\infty=\max_{i=1,\dots,N}\sum_j|a_{ij}|$ coincides with the operator norm $\|Q\|=\max_{\|x\|_1=1}\|Qx\|_1$ associated to $\|\cdot\|_1$  (see \cite{Horn-Johnson}, Definition 5.6.1 and Example 5.6.5). Also  if we define the operator $\psi(x)\doteqdot x^TQ$, then the inequality  $\|\psi x\|_1\le \|\psi\|\ \|x\|_1$ becomes $\|x^TQ\|_1\le \|x^T\|_1\|Q\|_\infty=\|x^T\|_1\|Q\|$. In particular, if $Q$ is row-stochastic this means that $\|x^TQ\|_1\le \|x^T\|_1$.\\
\end{rmk}

Now, we proceed with the \textbf{Zero-block column case}. In contrast to Theorem \ref{thm:diagonal-case}, in this case, the iteration of the personalization vector converges to a vector whose first $n_1$-entries are equal to zero, while the remaining $n_2$-entries are the left-hand Perron vector $\bold{c}_2\in\R^{n_2\times1}$ of the row-normalization $P_{A_2}$, which corresponds to the dangling cluster part of the network. More precisely,\\

\begin{thm}\label{thm:zero-case}
Let $A$ be a non-negative reducible square matrix of order $N$ of the form
\begin{equation*}
    A=\left(\begin{array}{c|c} 
    \bold{0} & B \\ 
    \hline
    \bold{0} & A_2 \end{array}\right),
\end{equation*}
where $A_2$ is a non-negative and irreducible square matrix of order $n_2$ and $B$ is a non-negative matrix of size $n_1\times n_2$, with $n_1+n_2=N$. Let $P_{B}$ and $P_{A_2}$ be the row-normalization of matrices $B$ and $A_2$, respectively, defined in Section~\ref{Section:notation}. Then for any personalization vector $\bold{v}\in\R^{N\times 1}$ with $\|\bold{v}\|_1=1$, the recursive sequence $\{\bold{x}_k\}_{k\geq0}\subset\R^{N\times1}$, with $\bold{x}_0=\bold{v}$, defined as 
\begin{equation*}
 \bold{x}_k^{T}=\bold{x}_{k-1}^{T}X(\lambda)=\bold{v}^TX(\lambda)^k\,,\quad (k\geq 1),
\end{equation*}
converges to the vector $\bold{x}^T=(\enspace\bold{0}_{n_1\times 1}^T \quad | \quad \bold{c}_2^T\enspace)$, where $\bold{c}_2\in\R^{n_2\times1}$ is the left-hand Perron vector of the matrix $P_{A_2}$.
\end{thm}

\begin{proof}
Let $P_A$ be the row-stochastic reducible matrix obtained from $A$ by row-normalization.
$$P_A=\left(\begin{array}{c|c}
\bold{0} & P_B \\ 
\hline
\bold{0} & P_{A_2} \end{array}\right).$$
 Notice that the ${1,2}$-block of $P_A$ coincides with $P_B$, the row-normalization of $B$, while  the ${2,2}$-block of $P_A$   coincides with $P_{A_2}$, the row-normalization of $A_2$,  and thus both $P_{B}$ and $P_{A_2}$ clearly are row-stochastic. As it was mentioned in Section \ref{Section:notation}, the Google matrix $G$ without dangling nodes defined as 
\begin{equation*}
G=\lambda P_A+(1-\lambda)\bold{ev}^T\in\mathbb{R}^{N\times N}
\end{equation*} 
is a row-stochastic and positive square matrix for which there exists a unique positive vector $\pi\in\R^{N\times1}$ (called PageRank vector), with $\pi^T\bold{e}=1$, satisfying the equation
\begin{equation}\label{eq:pagerank-zero-block-case}    \pi^T=(1-\lambda)\mathbf{v}^T(I_N-\lambda P_A)^{-1}=\mathbf{v}^TX(\lambda),
\end{equation}
where $X(\lambda)=(1-\lambda)(I_N-\lambda P_A)^{-1}$. A straightforward computation shows that the inverse of the reducible  matrix $I_N-\lambda P_A$ is the reducible matrix 
\begin{equation*}
(I_N-\lambda P_A)^{-1}
=\left(\begin{array}{c|c} 
I_{n_1} & \lambda P_B(I_{n_2}-\lambda P_{A_2})^{-1} \\
\hline
\bold{0} & (I_{n_2}-\lambda P_{A_2})^{-1} 
\end{array}\right).\\
\end{equation*}
Therefore, in this zero-block column case, equation (\ref{eq:pagerank-zero-block-case}) can be rewritten as follows
\begin{align*}
    \pi^T=\bold{v}^T X(\lambda)
    =\bold{v}^T(1-\lambda)(I_N-\lambda P_A)^{-1}
    &=\bold{v}^T(1-\lambda)\left(\begin{array}{c|c} 
    I_{n_1} & \lambda P_B(I_{n_2}-\lambda P_{A_2})^{-1} \\
    \hline
    \bold{0} & (I_{n_2}-\lambda P_{A_2})^{-1} 
    \end{array}\right)\\
    &=\bold{v}^T\left(\begin{array}{c|c}
    (1-\lambda)I_{n_1} & \lambda(1-\lambda)P_B(I_{n_2}-\lambda P_{A_2})^{-1} \\
   \hline
    \bold{0} & (1-\lambda)(I_{n_2}-\lambda P_{A_2})^{-1} 
    \end{array}\right)\\
    &=\bold{v}^T\left(\begin{array}{c|c} (1-\lambda)I_{n_1}  & \lambda P_B X_2(\lambda)  \\
    \hline
     \bold{0} & X_2(\lambda) \end{array}\right),
\end{align*}
where $X_2(\lambda) =(1-\lambda)(I_{n_2}-\lambda P_{A_2})^{-1}$. Now, let us consider a recursive sequence of vectors $\{\bold{x}_k^T\}_{k\geq0}$ defined by
\begin{equation}\label{eq:sequence-xk-zero-case}
 \bold{x}_k^{T}=\bold{x}_{k-1}^{T}X(\lambda)=\bold{x}_0^TX(\lambda)^k\,,\quad (k\geq 1)\,,\,\quad\text{with}\quad \, \bold{x}_0=\bold{v}>0\quad  \text{and}\quad \Vert \bold{v}\Vert_1=1.   
\end{equation}
Let $\bold{v}^T=(\enspace\bold{v}_1^T\quad|\quad \bold{v}_2^T\enspace)$ be a decomposition of vector $\bold{v}$ where $\bold{v}_1\in\R^{n_1\times1}$ and $\bold{v}_2\in\R^{n_2\times1}$.
Since
$$X^k(\lambda)
=\left(\begin{array}{c|c} 
(1-\lambda)I_{n_1}  & \lambda P_B X_2(\lambda)  \\
\hline
\bold{0} & X_2(\lambda) \end{array}\right)^k
=\left(\begin{array}{c|c} 
(1-\lambda)^k I_{n_1}  & \lambda\sum_{i=1}^{k} (1-\lambda)^{k-i} P_BX_2^{i}(\lambda)  \\
\hline
\bold{0} & X_2(\lambda)^k \end{array}\right),\quad (k\geq1).$$
we get
\begin{align*}
    \bold{x}_k^T=\bold{v}^TX^{k}(\lambda)
    &=(\enspace\bold{v}_1^T\quad |\quad \bold{v}_2^T \enspace)\left(
    \begin{array}{c|c} 
    (1-\lambda)^k I_{n_1}  & \lambda\sum_{i=1}^{k} (1-\lambda)^{k-i} P_BX_2^{i}(\lambda)  \\
    \hline
    \bold{0} & X_2(\lambda)^k \end{array}\right)\\\nonumber
     &\\
    &=\left(\enspace(1-\lambda)^k\bold{v}^T_1 \quad \Big{|} \quad \lambda \bold{v}^T_1\sum_{i=1}^{k} (1-\lambda)^{k-i} P_BX_2^{i}(\lambda) + \bold{v}^T_2 X^{k}_2(\lambda)\enspace\right),\quad (k\geq1).
\end{align*}
That is, the first $n_1$ entries of the  decomposition $\bold{x}_k^T=(\enspace(\bold{x}_k)^T_1\quad |\quad(\bold{x}_k)^T_2\enspace)$ are obtained from the expression $(\bold{x}_k)_1^T=(1-\lambda)^k\bold{v}^T_1$, while the remaining $n_2=N-n_1$ entries satisfy $(\bold{x}_k)_2^T=\lambda \bold{v}^T_1\sum_{i=1}^{k} (1-\lambda)^{k-i} P_BX_2^{i}(\lambda) + \bold{v}^T_2 X^{k}_2(\lambda)$.\\

Alternatively, the vector sequence in  (\ref{eq:sequence-xk-zero-case}) can be recursively expressed  as follows
\begin{align}
\label{eq:expression-pi-k-recursively-CASE2}\nonumber
    \bold{x}_{k+1}^T=\bold{x}_{k}^TX(\lambda)
    &=(\enspace(\bold{x}_k)_1^T\quad | \quad (\bold{x}_k)_2^T\enspace)\left(\begin{array}{c|c} 
    (1-\lambda)I_{n_1}  & \lambda P_B X_2(\lambda)\\
    \hline
    \bold{0} & X_2(\lambda) \end{array}\right)\\ \nonumber
    \\
    &=\left(\enspace(1-\lambda)(\bold{x}_k)_1^T \quad | \quad \lambda(\bold{x}_k)_1^TP_B X_2(\lambda) + (\bold{x}_k)_2^T X_2(\lambda)\enspace\right).
\end{align}
Observe that, since $\lambda\in (0,1)$, we have $\displaystyle\lim_{k\to\infty}(1-\lambda)^k\bold{v}^T_1=\bold{0}_{n_1\times1}^T$. Thus, for a given $\varepsilon>0$, there exists a positive integer $K\geq1$ such that for every $k\geq K$ the condition 
\begin{equation}\label{eq:bound1-zero-case}
\left\Vert (\bold{x}_k^T)_1\right\Vert_1=\left\Vert (1-\lambda)^k\bold{v}^T_1\right\Vert_1<\varepsilon   
\end{equation} holds. Then, from equation (\ref{eq:expression-pi-k-recursively-CASE2}), for every $k\geq K$ we can write the vector $\bold{x}_{k}^T$ as follows  
\begin{equation}\label{eq:recursive-expression}
\bold{x}_{k}^T=\bold{x}_{K}^TX^{k-K}(\lambda)=
\left(\enspace (1-\lambda)^{k-K}(\bold{x}_K)_1^T\quad \Big{|}\quad \lambda(\bold{x}_K)_1^T\displaystyle\sum_{i=1}^{k-K}(1-\lambda)^{k-K-i} P_BX_2^{i}(\lambda) + (\bold{x}_K)_2^T X_2^{k-K}(\lambda)\enspace\right).         
\end{equation}
That is, for $k\ge K$, the vector $\bold{x}_k^T$ can be written in terms of  $\bold{x}_K^T$. 
Now, we compute a bound for the $n_2$ components of the vector $\bold{x}_k^T$, that is, the first half of the right-hand part of the vector in equation (\ref{eq:recursive-expression}). Since $P_B$ and $X_{2}$ are row-stochastic matrices and $\lambda\in(0,1)$, we can use Remark \ref{rmk:operator-norm} to obtain 
\begin{align}\label{eq:bound2-zero-case}
\left\Vert \lambda(\bold{x}_K)^T_1\sum_{i=1}^{k-K}(1-\lambda)^{k-K-i} P_BX_2^{i}(\lambda)\right\Vert_1\nonumber
&\leq \lambda \left\Vert (\bold{x}_K)^T_1\right\Vert_1 \sum_{i=1}^{k-K} (1-\lambda)^{k-K-i} \Vert P_B\Vert \,\,\Vert X_2^{i}(\lambda)\Vert
=\lambda \left\Vert (\bold{x}_K)^T_1\right\Vert_1 \sum_{i=1}^{k-K}(1-\lambda)^{k-K-i}\\\nonumber
&=\lambda \left\Vert (\bold{x}_K)^T_1\right\Vert_1 \sum_{j=0}^{k-K-1} (1-\lambda)^j
\leq\lambda \left\Vert (\bold{x}_K)^T_1\right\Vert_1 \sum_{j=0}^{\infty} (1-\lambda)^j
=\lambda \left\Vert (\bold{x}_K)^T_1\right\Vert_1 \frac{1}{1-(1-\lambda)}\\
&=\lambda \left\Vert (\bold{x}_K)^T_1\right\Vert_1 \frac{1}{\lambda}=\left\Vert (\bold{x}_K)^T_1\right\Vert_1 <\varepsilon.
\end{align}
By Remark \ref{rmk:resolvent-isometry}, the sequence of vectors $\{\bold{x}_k^T\}_{k\geq0}$ defined in equation (\ref{eq:sequence-xk-zero-case}) satisfies    $\Vert \bold{x}_k\Vert_1=1$, for all $k\geq0$. 
In addition, from equation (\ref{eq:bound1-zero-case}) follows $1-\varepsilon<\Vert (\bold{x}_K)_2\Vert_1\leq 1$, since $X(\lambda)$ is row-stochastic. Now, for a fixed $\lambda\in (0,1)$ we apply Theorem \ref{thm:main-theorem} to the (non-negative irreducible) matrix $A_2$  and the index $K$ above
\begin{equation*}
    \lim_{k\to\infty}(\bold{x}_K)_2^TX_2^k(\lambda)=
    \Vert (\bold{x}_K)_2^T\Vert_1 \lim_{k\to\infty}\frac{(\bold{x}_K)_2^T}{\Vert (\bold{x}_K)_2^T\Vert_1}X_2^k(\lambda)=\Vert (\bold{x}_K)_2^T\Vert_1 \bold{c}_2^T,
\end{equation*}
where $\bold{c}_2\in\R^{n_2\times1}$ is the left-hand Perron vector for $P_{A_2}$. Therefore, for a given $\varepsilon>0$ and for the fixed vector $(\bold{x}_K)_2\in\R^{n_2\times1}$ with index $K$ is provided by the condition in equation (\ref{eq:bound1-zero-case}), there is a positive integer $L>1$ with $L>K$ such that for every $k>L$ we have 
\begin{equation}\label{eq:bound3-zero-case}
\left\Vert\enspace (\bold{x}_K)_2^T X_2^{k-K}(\lambda) -\Vert (\bold{x}_K)_2^T\Vert_1 \bold{c}_2^T\enspace \right\Vert_1<\varepsilon\qquad \text{with}\qquad 1-\varepsilon<\Vert (\bold{x}_K)_2\Vert_1\leq 1.
\end{equation}
In summary, combining equations (\ref{eq:bound1-zero-case}), (\ref{eq:bound2-zero-case}) and (\ref{eq:bound3-zero-case}) with the fact that $\lim_{\varepsilon\to0^+}\Vert (\bold{x}_K)_2\Vert_1=1$, we have proved that for every $\varepsilon>0$, there exists a positive integer $k$ large enough such that the estimates

\begin{equation}\label{Zerocase}
\left\Vert(\bold{x}_k^T)_1\right\Vert_1<\varepsilon \qquad \text{and}\qquad \left\Vert\enspace (\bold{x}_k)_2^T - \bold{c}_2^T\enspace \right\Vert_1<\varepsilon\quad
\end{equation} 
hold for the vector decomposition $\bold{x}_k^T=(\enspace (\bold{x}_k)_1^T\quad | \quad  (\bold{x}_k)_2^T\enspace)$. Therefore, the sequence $\{\bold{x}_k\}_{k\geq0}$ in equation (\ref{eq:sequence-xk-zero-case}) converges to a non-zero vector $\bold{x}\in\R^{N\times 1}$ with $\Vert\bold{x}\Vert_1=1$ such that
\begin{equation*}
    \bold{x}^T=\lim_{k\to\infty}\bold{v}^TX^{k}(\lambda)
    =\lim_{k\to\infty}\bold{x}_k^T
    =\lim_{k\to\infty}\left(\enspace (\bold{x}_k)_1^T\quad|\quad (\bold{x}_k)_2^T \enspace\right)
    =(\enspace \bold{0}_{n_1\times 1}^T \quad | \quad \bold{c}_2^T\enspace),
\end{equation*}
where $\bold{c}_2\in \R^{n_2\times1}$ is the left-hand Perron vector associated with the non-negative and irreducible square matrix $P_{A_2}$.

\end{proof}

In this case, Theorem~\ref{thm:diagonal-case} shows that if we consider a non-strongly connected directed graph $\mathcal{G}=(V,E)$ with a strongly connected component and some source nodes, then, for any damping factor $\lambda\in (0,1)$,  PageRank has no fixed points as a function $PR_\lambda:\,~\Delta_N^+\,~\to\,~\Delta_N^+$ , but it has a unique fixed point if we consider it as a function $PR_\lambda: \Delta_N\longrightarrow \Delta_N$ where
\[
\Delta_N=\left\{\bold{x}=(x_1,\cdots,x_N)^T\in\R^{N\times1}\,:\,\enspace x_1+\cdots+x_N=1,\, x_i\ge 0,\, {\text{for all }}1\le i\le N \,\right\},
\]
as the following result shows.\\

\begin{coro}\label{cor:source-connected}
Let $\mathcal{G}=(V,E)$ be a non-strongly connected directed graph  such that it has only one strongly connected component $\mathcal{G}_2$ and the rest nodes are source nodes. Then, for every damping factor $\lambda\in (0,1)$, $PR_\lambda: \Delta_N^+\longrightarrow \Delta_N^+$ has no fixed points but  $PR_\lambda: \Delta_N\longrightarrow \Delta_N$ has a unique fixed point. Furthermore, this fixed point is of the form $\bold{x}^T=(\enspace\bold{0}_{n_1\times 1}^T \quad | \quad \bold{c}_2^T\enspace)$, where $\bold{c}_2\in\R^{n_2\times1}$ is the left-hand Perron vector of the row-normalization square matrix $P_{A_2}$ of order $n_2$ associated to the component $\mathcal{G}_2$.\\
\end{coro}

Finally, we present the \textbf{General reducible case} in which, as in the previous result, the iteration of the personalization vector converges to the vector $(\enspace\bold{0}_{n_1\times 1}^T \quad | \quad \bold{c}_2^T\enspace)$, where $\bold{c}_2\in\R^{n_2\times1}$ is the left-hand Perron vector of the matrix $P_{A_2}$.  We remark that although the arguments to show that the first $n_1$-entries tend to zero are only slightly different from  the ones  in Theorem \ref{thm:zero-case}, the conclusion for the remaining entries requires quite a harder work. Before stating this result,  recall that for real matrices $A=(a_{ij})_{i,j}$ and $B=(b_{ij})_{i,j}$ of the same order, the inequality $A\geq B$ means that $a_{ij}\geq b_{ij}$ for all $i,j$.\\

Now, we are in position to state the convergence result for the \textbf{General reducible case}.\\

\begin{thm}\label{thm:non-zero-case}
Let $A$ be a non-negative reducible square matrix of order $N$ of the form
\begin{equation*}
    A=\left(\begin{array}{c|c} 
    A_1 & B \\ 
    \hline
    \bold{0} & A_2 \end{array}\right),
\end{equation*}
where $A_1$ and $A_2$ are non-negative irreducible square matrices of order $n_1$ and $n_2$, respectively, and $B$ is a non-negative, different from zero, matrix of size $n_1\times n_2$, with $n_1+n_2=N$. Let $P_{A}$ and $P_{A_2}$ be the row-normalization of matrices $A$ and $A_2$, respectively. 
For a fixed $\lambda\in (0,1)$, let consider $X(\lambda)=(1-\lambda)(I_N-\lambda P_A)^{-1}$. Then for any personalization vector $\bold{v}\in\R^{N\times 1}$ with $\|\bold{v}\|_1=1$, the recursive sequence $\{\bold{x}_k\}_{k\geq0}\subset\R^{N\times1}$, with $\bold{x}_0=\bold{v}$, defined as 
\begin{equation*}
 \bold{x}_k^{T}=\bold{x}_{k-1}^{T}X(\lambda)=\bold{v}^TX(\lambda)^k\,,\quad (k\geq 1),
\end{equation*}
converges to the vector $\bold{x}^T=(\enspace\bold{0}_{n_1\times1}^T \quad | \quad \bold{c}_2^T\enspace)$, where $\bold{c}_2\in\R^{n_2\times1}$ is the left-hand Perron vector of the matrix $P_{A_2}$.
\end{thm}

\begin{proof} 
Let $P_A$ be the row-stochastic reducible matrix obtained from $A$ by row-normalization. 
$$P_A=\left(\begin{array}{c|c} 
Q_{A_1} & Q_B \\
\hline
\bold{0} & P_{A_2} \end{array}\right).$$
Notice that $P_{A_2}$ is row-stochastic while matrices $Q_{A_1}$ and $Q_B$ are not, still they are  jointly row-stochastic, that is,
\begin{equation*}
\sum_{j=1}^N \left(\begin{array}{c|c}
Q_{A_1} & Q_B
\end{array}\right)
_{i,j}=\sum_{j=1}^{n_1}\left(Q_{A_1}\right)_{i,j}\,+\, \sum_{j=1}^{n_2}\left(Q_B\right)_{i,j}=1,\quad \text{for all } i=1,2,\dots, n_1.
\end{equation*}
As it was mentioned in Section \ref{Section:notation}, the Google matrix $G$ without dangling nodes defined as
\begin{equation}
G=\lambda P_A+(1-\lambda)\bold{ev}^T\in\mathbb{R}^{N\times N}
\end{equation} 
is a row-stochastic square matrix for which there exists a unique positive vector $\pi\in\R^{n\times1}$ (called PageRank vector), with $\pi^T\bold{e}=1$, satisfying the equation
\begin{equation}
    \pi^T=(1-\lambda)\mathbf{v}^T(I_N-\lambda P_A)^{-1}=\mathbf{v}^TX(\lambda),\label{eq:pagerank}
\end{equation}
where $X(\lambda)=(1-\lambda)(I_N-\lambda P_A)^{-1}$. As above, the inverse of the reducible square matrix $I_N-\lambda P_A$ is the reducible matrix 
\begin{equation*}
(I_N-\lambda P_A)^{-1}=\left(\begin{array}{c|c} 
(I_{n_1}- \lambda Q_{A_1}) & -\lambda Q_{B}\\
\hline 
\bold{0} & (I_{n_2}-\lambda P_{A_2}) 
\end{array}\right)^{-1}
=\left(\begin{array}{c|c} 
(I_{n_1}- \lambda Q_{A_1})^{-1} & \lambda(I_{n_1}-\lambda Q_{A_1})^{-1}Q_B(I_{n_2}-\lambda P_{A_2})^{-1} \\
\hline 
\bold{0} & (I_{n_2}-\lambda P_{A_2})^{-1} 
\end{array}\right). \\
\end{equation*}
Therefore, equation (\ref{eq:pagerank}) can be rewritten as follows
\begin{align*}
    \pi^T=\bold{v}^T X(\lambda)
    &=\bold{v}^T(1-\lambda)(I_N-\lambda Q_A)^{-1}\\
    &=\bold{v}^T(1-\lambda)\left(\begin{array}{c|c}
(I_{n_1}- \lambda Q_{A_1})^{-1} & \lambda(I_{n_1}-\lambda Q_{A_1})^{-1}Q_B(I_{n_2}-\lambda P_{A_2})^{-1} \\
\hline
\bold{0} & (I_{n_2}-\lambda P_{A_2})^{-1} 
\end{array}\right)\\
    &=\bold{v}^T\left(\begin{array}{c|c}
(1-\lambda)(I_{n_1}- \lambda Q_{A_1})^{-1} & \lambda(1-\lambda)(I_{n_1}-\lambda Q_{A_1})^{-1}Q_B(I_{n_2}-\lambda P_{A_2})^{-1} \\
\hline
\bold{0} & (1-\lambda)(I_{n_2}-\lambda P_{A_2})^{-1} 
\end{array}\right)\\
    &= \bold{v}^T\left(\begin{array}{c|c} 
    X_1(\lambda)  & \lambda X_1(\lambda)Q_B(I-\lambda P_{A_2})^{-1}  \\
     \hline
     \bold{0} & X_2(\lambda) \end{array}\right)
     =\bold{v}^T\left(\begin{array}{c|c}
     X_1(\lambda)  & \lambda(1-\lambda)^{-1}X_1(\lambda)Q_B X_2(\lambda)  \\
     \hline
     \bold{0} & X_2(\lambda) \end{array}\right)
\end{align*}
where $X_1(\lambda) =(1-\lambda)(I_{n_1}- \lambda Q_{A_1})^{-1}$ (the existence of which if justified as in Proof of Step 1 below) and $X_2(\lambda) =(1-\lambda)(I_{n_2}-\lambda P_{A_2})^{-1}$. Now, let us consider a recursive sequence of vectors $\{\bold{x}_k^T\}_{k\geq0}$ defined by 
\begin{equation}\label{eq:sequence-xk-non-zero-case}
  \bold{x}_k^{T}=\bold{x}_{k-1}^{T}X(\lambda)\,,\quad(k\geq 1)\,,\,\qquad\text{with}\qquad \, \bold{x}_0=\bold{v}>0 \qquad\text{and}\qquad\Vert \bold{v}\Vert_1=1.   
\end{equation}
Let $\bold{v}^T=(\enspace\bold{v}_1^T\quad | \quad \bold{v}_2^T\enspace)$ be a decomposition of vector $\bold{v}$ where $\bold{v}_1\in\R^{n_1\times1}$ and  $\bold{v}_2\in\R^{n_2\times1}$. Since
$$X^k(\lambda)
=\left(\begin{array}{c|c} 
X_1(\lambda)  & \lambda(1-\lambda)^{-1}X_1(\lambda)Q_B X_2(\lambda)\\
 \hline
\bold{0} & X_2(\lambda) \end{array}\right)^k
=\left(\begin{array}{c|c} X_1(\lambda)^k  & \lambda(1-\lambda)^{-1}\sum_{i=1}^{k} X^{k-i+1}_1(\lambda) Q_BX_2^{i}(\lambda)  \\
\hline
\bold{0} & X_2(\lambda)^k \end{array}\right), \enspace k\geq1,$$
we have
\begin{align*}
    \bold{x}_k^T=\bold{v}^TX^{k}(\lambda)
    &=(\enspace\bold{v}_1^T\quad |\quad \bold{v}_2^T \enspace)\left(\begin{array}{c|c} X_1(\lambda)^k  & \lambda(1-\lambda)^{-1}\sum_{i=1}^{k} X^{k-i+1}_1(\lambda) Q_BX_2^{i}(\lambda)  \\
    \hline
    \bold{0} & X_2(\lambda)^k \end{array}\right)\\\nonumber
     &\\
    &=\left(\enspace\bold{v}^T_1X_1(\lambda)^k \quad | \quad \lambda(1-\lambda)^{-1} \bold{v}^T_1\sum_{i=1}^{k} X^{k-i+1}_1(\lambda) Q_BX_2^{i}(\lambda) + \bold{v}^T_2 X^{k}_2(\lambda)\enspace\right).
\end{align*}
That is, in the decomposition $\bold{x}_k^T=(\enspace(\bold{x}_k)^T_1\quad |\quad(\bold{x}_k)^T_2\enspace)$, the first $n_1$ entries of the vector $\bold{x}_k^T$ satisfy $(\bold{x}_k)_1^T=\bold{v}^T_1X_1(\lambda)^k$, while the remaining  $n_2=N-n_1$ entries  satisfy $(\bold{x}_k)_2^T=\lambda(1-\lambda)^{-1} \bold{v}^T_1\sum_{i=1}^{k} X^{k-i+1}_1(\lambda) Q_BX_2^{i}(\lambda) + \bold{v}^T_2 X^{k}_2(\lambda)$.\\

Alternatively, the vector sequence defined in equation (\ref{eq:sequence-xk-non-zero-case}) can be recursively expressed  as follows
\begin{align}
\label{eq:expression-pi-n-recursively}\nonumber
    \bold{x}_{k+1}^T
    &=\bold{x}_{k}^TX(\lambda)=(\enspace(\bold{x}_k)_1^T\quad |\quad  (\bold{x}_k)_2^T\enspace)
    \left(\begin{array}{c|c} 
    X_1(\lambda)  & \lambda(1-\lambda)^{-1}X_1(\lambda)Q_B X_2(\lambda)\\
    \hline
    \bold{0} & X_2(\lambda) 
     \end{array}\right)\\\nonumber
     \\
     &=\left(\enspace(\bold{x}_k)_1^T X_1(\lambda)\quad | \quad \lambda(1-\lambda)^{-1}(\bold{x}_k)_1^TX_1(\lambda)Q_B X_2(\lambda) + (\bold{x}_k)_2^T X_2(\lambda)\enspace\right)
\end{align}

With equation (\ref{eq:expression-pi-n-recursively}) at hand, the idea is to prove that the sequence $\{\bold{x}_k^T\}_{k\geq0}$ defined in equation (\ref{eq:sequence-xk-non-zero-case}) converges to the vector $\bold{x}^T=(\enspace \bold{0}_{n_1\times1}^T\quad | \quad \bold{c}_2^T\enspace)$, where $\bold{c}_2\in \R^{n_2\times1}$ is the left-hand Perron vector associated with the non-negative and irreducible square matrix $P_{A_2}$.\\

The proof of this fact is based on the following two steps:
\begin{itemize}
    \item [\textbf{Step1.}] If we consider the decomposition $\bold{x}_k^T=(\enspace(\bold{x}_k)^T_1\quad | \quad (\bold{x}_k)^T_2\enspace)$ with $(\bold{x}_k)_1\in\R^{n_1\times 1}$, then $\displaystyle\lim_{k\to\infty}(\bold{x}_k)^T_1=\bold{0}_{n_1\times1}^T$.\\Equivalently, $\displaystyle\lim_{k\to\infty}\bold{v}^T_1 X_1^k(\lambda)=\bold{0}_{n_1\times1}^T$. 
    \item [\textbf{Step2.}] For every $\varepsilon>0$, there exists a positive integer $K\geq1$ (from \textbf{Step 1}) such that for every integer $k\geq K$ the following condition 
    $$\displaystyle\left\Vert \frac{\lambda}{1-\lambda}(\bold{x}_K)^T_1\sum_{i=1}^{k-K}X_1^{k-K-i+1}(\lambda) Q_BX_2^{i}(\lambda)\right\Vert_1<\varepsilon$$
    is fulfilled.
\end{itemize}

We remark that these two steps basically align with those in Theorem \ref{thm:zero-case}, although the arguments employed here are  more sophisticated.\\

\underline{\textbf{Proof of Step 1}}: It is clear that we can  perturb the irreducible matrix $Q_{A_1}$ with some positive $E\ge 0$ to produce a matrix $P=Q_{A_1}+E=(p_{ij})$ 
satisfying:
\begin{itemize}
    \item[(C1)] $P$ is row-stochastic, and
    \item[(C2)] $0\leq Q_{A_1} \leq P$ with $Q_{A_1} \neq P$.
\end{itemize}

Firstly, observe that $\rho(P)=1$ is a consequence of the inequality $\min_j\sum_{i}p_{ij}\leq\rho(P)\leq\max_j\sum_{i}p_{ij}$ and condition (C1). Now, since  $Q_{A_1}$ is a non-negative irreducible square matrix and $Q_{A_1} \neq P$ we have $\rho(Q_{A_1})<\rho(P)=1$ (see \cite[Corollary 3.3.29]{BemannPlemmons}). We next show that $X_1(\lambda)>0$ is irreducible and $X_1(\lambda)
=(1-\lambda)(I_{n_1}-\lambda Q_{A_1})^{-1}\lneq (1-\lambda)(I_{n_1}-\lambda P)^{-1}$. Indeed, since $Q_{A_1}\geq0$ is irreducible and $\rho(Q_{A_1})<1$, we have 
$X_1(\lambda)
>0$ for $0<\lambda<1$ (see \cite[Theorem 3, Section 3, Chapter XIII]{Gantmacher}). Moreover, since its Neumann expansion  is 
\begin{equation}\label{eq:X_1_tilde}
X_1(\lambda)
=(1-\lambda)(I_{n_1}- \lambda Q_{A_1})^{-1}
=(1-\lambda)\sum_{i=0}^{\infty}\left(\lambda Q_{A_1}\right)^i
=(1-\lambda)\left[I_{n_1}+(\lambda Q_{A_1})+ (\lambda Q_{A_1})^2+\dots\right], 
\end{equation}
as in Lemma \ref{lemma:resolvent-row-stochastic}, the series in equation (\ref{eq:X_1_tilde}) is irreducible since it is the sum an irreducible matrix $Q_{A_1}$ and a non-negative matrix (see \cite[Theorem 1]{SchwarzB}). The inequality $X_1(\lambda)\lneq (1-\lambda)(I_{n_1}-\lambda P)^{-1}$ trivially follows from the Neumann expansions and $Q_{A_1}\lneq P$ in condition (C2).

Moreover, since $(1-\lambda)(I_{n_1}- \lambda P)^{-1}$ is positive and row-stochastic (by Lemma \ref{lemma:resolvent-row-stochastic}) we conclude that the  spectral radius $\rho(\,(1-\lambda)(I_{n_1}- \lambda P)^{-1}\,)=1$, which implies thethe strict inequality $\rho(X_1(\lambda))<\rho(\,(1-\lambda)(I_{n_1}- \lambda P)^{-1}\,)=1$ (see \cite[Corollary 3.3.29]{BemannPlemmons}). Therefore, we conclude that $\displaystyle\lim_{k\to\infty}X_1^k(\lambda)=0$ (see \cite[Theorem 5.6.12]{Horn-Johnson}) and \textbf{Step 1} is proved.\\

At this point, observe that \textbf{Step 1} implies that for the sequence $\{\bold{x}_k^T\}_{k\geq0}$ defined recursively by equation (\ref{eq:sequence-xk-non-zero-case}) and a given $\varepsilon>0$, there exists a positive integer $K\geq1$ such that for every $k\geq K$ the condition 
\begin{equation}\label{eq:bound1-non-zero-case}
\Vert(\bold{x}_k)^T_1 \Vert_1<\varepsilon  
\end{equation}
is fulfilled.

Before proceeding with the proof of \textbf{Step 2}, let us provide some details regarding the upper bound for the norm of the matrix $\sum_{i=1}^{k-K}X_1^{k-K-i+1}(\lambda) Q_BX_2^{i}(\lambda)$. Firstly, by Remark \ref{rmk:operator-norm}, it is clear that $\Vert Q_B\Vert\leq1 $. On the other hand, since $P_{A_2}$ is row-stochastic with $\rho(P_{A_2})=1$, for fixed $\lambda\in(0,1)$, the matrix $X_2(\lambda)=(1-\lambda)(I_{n_2}-\lambda P_{A_2})^{-1}$ exists with Neumann expansion
$$X_2(\lambda)=(1-\lambda)(I_{n_2}-\lambda P_{A_2})^{-1}=(1-\lambda)\sum_{i=0}^\infty\,(\lambda P_{A_2})^i=(1-\lambda)\sum_{i=0}^\infty\,\lambda^i P_{A_2}^i,$$
which implies that $\Vert X_2(\lambda)\Vert\leq 1$. Regarding $X_1(\lambda)$, notice that from the proof of \textbf{Step 1} follows that $\rho(X_1(\lambda))<1$. Therefore, for a given $\varepsilon>0$, there exists a matrix norm $\Vert \cdot\Vert_\rho$ such that 
\begin{equation}
\label{eq:upper-bound-X-1-tilde}
\rho(X_1(\lambda))\leq \Vert X_1(\lambda)\Vert_\rho\leq \rho(X_1(\lambda))+\varepsilon\text{\qquad with\qquad }\rho(X_1(\lambda))+\varepsilon<1,
\end{equation}
(see \cite[Lemma 5.6.10]{Horn-Johnson})). Equation (\ref{eq:upper-bound-X-1-tilde}) will be crucial in the proof of \textbf{Step 2}.\\ 

Finally, recall that any two norms defined on a finite-dimensional space are equivalent (see \cite[Corollary 434]{Gockenbach}). In our case, for square matrices of order $n_1$, there exist positive constants $a,b\in\R$ with $0<a<b$ such that  
\begin{equation}
\label{eq:norm-equivalence}
   a\Vert \cdot\Vert\leq \Vert \cdot\Vert_\rho\leq b\Vert \cdot\Vert.
\end{equation}

With these ingredients at hand, we continue with the the proof of \textbf{Step 2}.\\

\underline{\textbf{Proof of Step 2}}: From \textbf{Step 1} and the argument of the preceding paragraph, for a given $\varepsilon>0$, there exist:

\begin{itemize}
    \item[(C3)] A norm matrix $\Vert \cdot\Vert_\rho$ such that $\rho(X_1(\lambda))\leq \Vert X_1(\lambda)\Vert_\rho\leq \rho(X_1(\lambda))+\varepsilon$ with 
$\rho(X_1(\lambda))+\varepsilon\leq\rho(X_1(\lambda))^{1/2}<1$, and
    \item[(C4)] A positive integer $K\geq 1$ such that $\Vert(\bold{x}_k)^T_1 \Vert_1<\varepsilon a(1-\rho(X_1(\lambda))^{1/2})(1-\lambda)/\lambda$ for every $k\geq K$, where the vector decomposition $\bold{x}_k^T=(\enspace(\bold{x}_k)^T_1\quad |\quad (\bold{x}_k)^T_2\enspace)$ with $(\bold{x}_k)_1\in\R^{n_1\times 1}$ is considered and $a>0$ is the constant in equation (\ref{eq:norm-equivalence}).
\end{itemize}

From equation (\ref{eq:expression-pi-n-recursively}), for every $k\geq K$ we can write the vector $\bold{x}_{k}^T$ as follows  
\begin{equation}\label{eq:recursive-expression-general}
\bold{x}_{k}^T=\bold{x}_{K}^TX^{k-K}(\lambda)=\left(\enspace(\bold{x}_K)_1^T X^{k-K}_1(\lambda)\quad |\quad \lambda(1-\lambda)^{-1}(\bold{x}_K)_1^T\displaystyle\sum_{i=1}^{k-K}X_1^{k-K-i+1}(\lambda) Q_BX_2^{i}(\lambda) + (\bold{x}_K)_2^T X_2^{k-K}(\lambda)\enspace\right).         
\end{equation}
That is, the vector $\bold{x}_k^T$ can be written in terms of  $\bold{x}_K^T$ for the fixed index $K$ provided by condition (C4). Now, we compute a bound for the $n_2$ components of the vector $\bold{x}_k^T$, that is, the first half of the right-hand part of the vector in equation (\ref{eq:recursive-expression-general}).
Since $\Vert Q_B\Vert\leq 1$ and $\Vert X_2(\lambda) \Vert\leq 1$ we can use norm matrix equivalence in equation (\ref{eq:norm-equivalence}) together with Remark \ref{rmk:operator-norm}  to obtain
\begin{align}
\label{eq:bound2-non-zero-case}
\left\Vert \frac{\lambda}{1-\lambda}(\bold{x}_K)^T_1\sum_{i=1}^{k-K}X_1^{k-K-i+1}(\lambda) Q_BX_2^{i}(\lambda)\right\Vert_1\nonumber
&\leq \frac{\lambda}{1-\lambda}\left\Vert (\bold{x}_K)^T_1\right\Vert_1 \sum_{i=1}^{k-K} \Vert X_1^{k-K-i+1}(\lambda) \Vert\enspace\Vert Q_B\Vert\enspace\Vert X_2^{i}(\lambda)\Vert\\\nonumber
&\leq \frac{\lambda}{1-\lambda}\left\Vert (\bold{x}_K)^T_1\right\Vert_1 \sum_{i=1}^{k-K} \frac{1}{a}\Vert X_1^{k-K-i+1}(\lambda)\Vert_\rho\\\nonumber
&\leq \frac{\lambda}{1-\lambda}\left\Vert (\bold{x}_K)^T_1\right\Vert_1 \sum_{j=1}^{k-K} \frac{1}{a}\Vert X_1^{j}(\lambda)\Vert_\rho\\\nonumber
&\leq \frac{\lambda}{1-\lambda}\left\Vert (\bold{x}_K)^T_1\right\Vert_1 \frac{1}{a} \sum_{j=1}^{k-K} \Vert X_1(\lambda)\Vert^j_\rho\\\nonumber
&\leq \frac{\lambda}{1-\lambda}\left\Vert (\bold{x}_K)^T_1\right\Vert_1 \frac{1}{a} \sum_{j=1}^{k-K}\left(\,\rho(X_1(\lambda))+\varepsilon\,\right)^j\\\nonumber
&\leq \frac{\lambda}{1-\lambda}\left\Vert (\bold{x}_K)^T_1\right\Vert_1 \frac{1}{a} \sum_{j=0}^{\infty}\left(\,\rho(X_1(\lambda))\,\right)^{j/2}\\
&=\frac{\lambda}{1-\lambda}\frac{\left\Vert (\bold{x}_K)^T_1\right\Vert_1}{a}\cdot\frac{1}{1-\rho(X_1(\lambda))^{1/2}}<\varepsilon,
\end{align}
where the last two inequalities follow from conditions (C3) and (C4).\\

By Remark \ref{rmk:resolvent-isometry}, the sequence of vectors $\{\bold{x}_k^T\}_{k\geq0}$ defined in equation (\ref{eq:sequence-xk-non-zero-case}) satisfies  $\Vert \bold{x}_k\Vert_1=1$, for all $k\geq0$. As above, the estimate in equation  (\ref{eq:bound1-non-zero-case}) and the row-stochasticity of $X(\lambda)$ imply that $1-\varepsilon<\Vert (\bold{x}_K)_2\Vert_1\leq 1$. Now, we proceed as in the proof of Theorem \ref{thm:zero-case}. For $\lambda\in (0,1)$ fixed and $K$ as in  condition (C4), we apply Theorem \ref{thm:main-theorem} to the (non-negative irreducible) matrix  $A_2$ and obtain
\begin{equation*}
    \lim_{k\to\infty}(\bold{x}_K)_2^TX_2^k(\lambda)=
    \left\Vert (\bold{x}_K)_2^T\right\Vert_1 \lim_{k\to\infty}\frac{(\bold{x}_K)_2^T}{\Vert (\bold{x}_K)_2^T\Vert_1}X_2^k(\lambda)=\Vert (\bold{x}_K)_2^T\Vert_1 \bold{c}_2^T,
\end{equation*}
where $\bold{c}_2\in\R^{n_2\times1}$ is the left-hand Perron vector of $P_{A_2}$. Therefore, for a given $\varepsilon>0$ and for the fixed vector $(\bold{x}_K)_2\in\R^{n_2\times1}$, there is a positive integer $L>1$ with $L>K$ such that for every $k>L$ we have 
\begin{equation}\label{eq:bound3-non-zero-case}
\left\Vert\enspace (\bold{x}_K)_2^T X_2^{k-K}(\lambda) -\Vert (\bold{x}_K)_2^T\Vert_1 \bold{c}_2^T\enspace \right\Vert_1<\varepsilon\qquad \text{and}\qquad 1-\varepsilon<\Vert (\bold{x}_K)_2\Vert_1\leq 1.
\end{equation}


In summary, combining \textbf{Step 1}, \textbf{Step 2} and equation (\ref{eq:bound3-non-zero-case}) 
we have proved (essentially) that for every $\varepsilon>0$, there exists a positive integer $K$ large enough such that for every $k\geq K$ we have
\begin{equation}\label{Generalcase}
\left\Vert(\bold{x}_k^T)_1\right\Vert_1<\varepsilon \qquad \text{and}\qquad \left\Vert\enspace (\bold{x}_k)_2^T - \Vert (\bold{x}_K)_2^T\Vert_1 \bold{c}_2^T\enspace \right\Vert_1<\varepsilon 
\end{equation}
where $\bold{x}_k^T=(\enspace (\bold{x}_k)_1^T\quad | \quad  (\bold{x}_k)_2^T\enspace)$. Therefore, the recursive sequence $\{\bold{x}_k\}_{k\geq0}$ defined in equation (\ref{eq:sequence-xk-non-zero-case}) converges to a non-zero vector $\bold{x}\in\R^{N\times 1}$ with $\Vert \bold{x}\Vert_1=1$,
\begin{equation*}
    \bold{x}^T=\lim_{k\to\infty}\bold{v}^TX^{k}(\lambda)
    =\lim_{k\to\infty}\bold{x}_k^T
    =\lim_{k\to\infty}\left(\enspace (\bold{x}_k)_1^T\quad|\quad (\bold{x}_k)_2^T \enspace\right)
    =(\enspace \bold{0}_{n_1\times1}^T \quad | \quad \bold{c}_2^T\enspace),
\end{equation*}
where $\bold{c}_2\in \R^{n_2\times1}$ is the left-hand Perron vector associated with the non-negative and irreducible square matrix $P_{A_2}$.\\

\end{proof}

Note that in this case, Theorem~\ref{thm:non-zero-case} proves that if we consider a non-strongly connected directed graph $\mathcal{G}=(V,E)$ with two interconnected strongly connected components, then, similarly to Corollary~\ref{cor:source-connected}, for a given damping factor $\lambda\in (0,1)$, the PageRank has no fixed points as a function $PR_\lambda: \Delta_N^+\longrightarrow \Delta_N^+$ but it has a unique fixed point when  considered as a function $PR_\lambda: \Delta_N\longrightarrow \Delta_N$, where
\[
\Delta_N=\left\{\bold{x}=(x_1,\cdots,x_N)^T\in\R^{N\times1}\,:\enspace x_1+\cdots+x_N=1,\, x_i\ge 0,\, {\text{for all }}1\le i\le N \,\right\}.
\]

\begin{coro}\label{cor:twoConnectedComponents}
Let $\mathcal{G}=(V,E)$ be a non-strongly connected directed graph  with two strongly connected components $\mathcal{G}_1$, $\mathcal{G}_2$ such that some nodes of $\mathcal{G}_1$ are connected to some nodes of $\mathcal{G}_2$, but no node of $\mathcal{G}_2$ is connected to any node in $\mathcal{G}_1$. Then, for every damping factor $\lambda\in (0,1)$, the operator $PR_\lambda: \Delta_N^+\longrightarrow \Delta_N^+$ has no fixed points but  $PR_\lambda: \Delta_N\longrightarrow \Delta_N$ has a unique fixed point. Furthermore, this fixed point is of the form $\bold{x}^T=(\enspace\bold{0}_{n_1\times 1}^T \quad | \quad \bold{c}_2^T\enspace)$, where $\bold{c}_2\in\R^{n_2\times1}$ is the row-normalization square matrix $P_{A_2}$ of order $n_2$ associated to the component $\mathcal{G}_2$.\\
\end{coro}

In the next section, we will use the results of Section \ref{Section:reducible} to investigate the convergence of the iteration of the PageRank vector under the action of the matrix $X(\lambda)$ in the case that the network $\mathcal{G}=(V,E)$ is not strongly connected and has  dangling clusters (recall that a dangling cluster is a group of nodes in $V$ with no outgoing links to other clusters within the network). In this setting, we will show that the iterations of the PageRank vector under the matrix $X(\lambda)$ converges, in the spirit of Theorem \ref{thm:diagonal-case}, to a vector which only depends on the left-hand Perron vector of the dangling clusters in the network.


\section{The case of non-strongly connected networks with dangling clusters}
\label{Section:dangling-cluster}
Let $P_A$ be the row-stochastic matrix obtained via row-normalization from some reducible square matrix $A$ of order $N$. The normal form of $P_A$ (see \cite[Section 2.3]{Varga})) is obtained after composing with an $N\times N$ permutation square matrix $S$ as follows
\begin{equation}
\label{eq:permutation-matrix}
    SP_AS^T = \left(\begin{matrix} 
    R_{1,1} & R_{1,2} & \dots  &R_{1,m} \\
    \bold{0}      & R_{2,2} & \dots  &R_{2,m}\\
    \vdots & \vdots & \ddots & \vdots\\
    \bold{0}      &  \bold{0}     &\dots   & R_{m,m}
    \end{matrix}\right),
\end{equation}
where each submatrix $R_{j,j}$, $1\leq j\leq m$, is either irreducible or a null matrix of order $1$.\\

As it was mentioned in Figure \ref{fig:dangling-cluster} of Section \ref{Section:reducible}, we define  a \textit{dangling cluster} as a group of nodes with no out-links to any other cluster of the network. In terms of the permutation matrix in equation (\ref{eq:permutation-matrix}), if $R_{i,i}$ is associated to a dangling cluster, then $R_{i,i}$ is a non-zero square matrix of order $n_i$ (the number of nodes inside this cluster) and $R_{i,j}=0$ for all $j=1,\dots,m$ with $j\neq i$.\\ 

In our case, let $\mathcal{G}=(V,E)$ be a directed graph where $V=\{1,2,\dots,N\}$ is the set of nodes. Let $M\geq1$ be the number of dangling clusters of the directed graph $\mathcal{G}$, where the connections among their nodes are given by matrices $D_i's$. If $L\geq1$ denotes the number of non-dangling clusters in $\mathcal{G}$, we can permute rows and columns of the matrix $P_A$, grouping together the nodes belonging to the same connected component and listing the dangling clusters in the last rows. Therefore, there exists a permutation square matrix $S$ of order $N$ such that $P_A$ can be expressed in the reduced form
\renewcommand{\arraystretch}{1.5}
\NiceMatrixOptions{columns-width=13mm,margin=0.9em}
\begin{equation}\label{eq:dangling-cluster} 
SP_AS^T =\begin{pNiceMatrix}
    \Block[draw]{}{Q_{1,1}}     & Q_{1,2}                   & \dots     &Q_{1,L}                    &Q_{1,L+1}                  &Q_{1,L+2}                  &\dots      &Q_{1,L+M}\\
    \bold{0}                    & \Block[draw]{}{Q_{2,2}}   & \dots     &Q_{2,L}                    &Q_{2,L+1}                  &Q_{2,L+2}                  &\dots      &Q_{2,L+M}\\
    \vdots                      & \vdots                    &\ddots     & \vdots                    &\vdots                     &\vdots                     &\ddots     &\vdots\\
    \bold{0}                    &  \bold{0}                 &\dots      & \Block[draw]{}{Q_{L, L}}  &Q_{L ,L+1}                 &Q_{L ,L+2}                 &\dots      &Q_{L,L+M}\\
    \bold{0}                    &  \bold{0}                 &\dots      & \bold{0}                  &\Block[draw]{}{D_{1}}      &\bold{0}                   &\dots      &\bold{0}\\
    \bold{0}                    &  \bold{0}                 &\dots      & \bold{0}                  &\bold{0}                   &\Block[draw]{}{D_{2}}      &\dots      &\bold{0}\\
    \bold{0}                    &  \bold{0}                 &\dots      & \bold{0}                  &\bold{0}                   &\bold{0}                   &\dots      &\bold{0}\\
    \vdots                      & \vdots                    &           & \vdots                    &\vdots                     &\vdots                     &\ddots     &\vdots\\
    \bold{0}                    &  \bold{0}                 &\dots      & \bold{0}                  &\bold{0}                   &\bold{0}                   &\dots      &\Block[draw]{}{D_M}\\
\end{pNiceMatrix},
\end{equation}
where the diagonal blocks $Q_{i,i}$, $i=1,2,\dots,L$, and $D_j$, $j=1,2,\dots,M$, are either irreducible or a null matrix of order $1$. In fact, due to the lack of dangling nodes, the block matrices $D_j$ cannot be a null matrix of order $1$ for any $j=1,2,\dots,M$. For simplicity of notation, we simply write the reduced form above $SP_AS^T$ as $P_A$.\\

With this reduced form for the row-normalization matrix $P_A$ in hand, we proceed to state the result for the case of a reducible matrix $P_A$ with dangling clusters.\\

\begin{thm}\label{thm:genral-case}
Let $A$ be a non-negative reducible square matrix of order $N$ and let denote by $P_A$ its row-normalization. Assume that $P_A$ is written in the reduced form given by equation (\ref{eq:dangling-cluster}), where the diagonal block matrices $Q_{i,i}$ are of order $n_i$, for $i=1,2,\dots,L$, and $D_j$ corresponds to the $j^{th}$-dangling cluster of order $m_j$, for $j=1,2,\dots,M$. For a fixed $\lambda\in (0,1)$, consider $X(\lambda)=(1-\lambda)(I_N-\lambda P_A)^{-1}$. Then for any personalization vector $\bold{v}\in\R^{N\times 1}$ with $\|\bold{v}\|_1=1$, the recursive sequence $\{\bold{x}_k\}_{k\geq0}\subset\R^{N\times1}$, with $\bold{x}_0=\bold{v}$, defined as 
\begin{equation*}
 \bold{x}_k^{T}=\bold{x}_{k-1}^{T}X(\lambda)=\bold{v}^TX(\lambda)^k\,,\quad (k\geq 1),
\end{equation*}
converges to the vector 
$$\bold{x}^T=(\enspace \bold{0}_{n_1\times1}^T\quad|\quad \bold{0}_{n_2\times1}^T\quad|\dots|\quad \bold{0}_{n_L\times1}^T \quad|\quad \alpha_1\bold{c}_1^T\,\,\quad |\quad \alpha_2\bold{c}_2^T\,\,\quad|\dots|\quad \alpha_M\bold{c}_M^T\enspace),$$
for some $\alpha_1,\alpha_2,\dots, \alpha_M \geq0$ with $\alpha_1+\alpha_2+\dots+\alpha_M=1$ and where $\bold{c}_1\in\R^{m_1\times1}, \bold{c}_2\in\R^{m_2\times1},\dots,\bold{c}_M\in\R^{m_M\times1}$ are the left-hand Perron vectors of the row-normalized  irreducible  matrices $D_1, D_2,\dots, D_M$, respectively.
\end{thm}

\begin{proof}

Since the proof of this result is based on a recursive argument, let us consider the sequence of matrices $\{P_{i,i}\}_{i\geq1}$ defined by 
\begin{equation}\label{eq:block-decomposition}
    P_{i,i}=\left(\begin{array}{l|c c c} 
	Q_{i,i} &  & Q_{B_i} &\\ 
	\hline 
    \\
     \scaleto{\bold{0}}{10pt}&  &\scaleto{P_{i+1,i+1}}{15pt} & \\
     \\
\end{array}\right)=
\left(\begin{array}{l|c c c c c c }
Q_{i,i}     & Q_{i,i+1}     &\dots      &Q_{i,L}    &Q_{i,L+1}      &\dots    &Q_{i,L+M}\\
\hline 
\bold{0}    &Q_{i+1,i+1}    &\dots      &Q_{i+1,L}  &Q_{i+1,L+1}    &\dots    &Q_{i+1,L+M} \\

\vdots      &\vdots         &\ddots     &\vdots     &\vdots         &\dots    &\vdots\\
\bold{0}    &\bold{0}       &\dots      &Q_{L,L}    &Q_{L,L+1}      &\dots    &Q_{L, L+M} \\
\bold{0}    &\bold{0}       &\dots      &\bold{0}   &D_{1}          &\dots    &\bold{0} \\
\vdots      &\vdots         &\dots      &\vdots     &\vdots         &\ddots   &\bold{0} \\
\bold{0}    &\bold{0}       &\dots      &\bold{0}   &\bold{0}       &\dots    &D_M \\
\end{array}\right),\quad (i\geq1),  
\end{equation}
with $P_{1,1}=P_A$, $Q_{B_i}\neq0$ and $Q_{i,i}$ is a non-negative and irreducible square matrix of order $n_i$ or a null matrix of order $1$, for $i=1,2\dots,L$. As in the proof of Theorem \ref{thm:non-zero-case} notice  that $Q_{B_i}$ is not necessarily row-stochastic, but $Q_{i,i}$ and $Q_{B_i}$ are jointly row-stochastic in the following sense
\begin{equation*}
\sum_{t=s}^N \left(\begin{array}{c|c}
Q_{i,i} & Q_{B_i}
\end{array}\right)
_{s,t}=1
,\quad \text{for all $s$ running in the appropriate rows}.
\end{equation*}
 On the other hand, the matrix $P_{i+1,i+1}$ is row-stochastic for $i=1,2,\dots,L$. 

Observe that, with the reduced form in equation (\ref{eq:block-decomposition}), the resolvent matrix $X(\lambda)=(1-\lambda)(I_N-\lambda P_A)^{-1}$ can be recursively described  as follows
    \begin{equation}
    \label{eq:recursive-resolvent}
    X_{j}(\lambda)=\left(\begin{array}{c|c c c} 
	(1-\lambda)(I_{n_j}-\lambda Q_{j,j})^{-1}  &  &\lambda (I_{n_j}-\lambda Q_{j,j})^{-1}Q_{B_j}X_{j+1}(\lambda)&\\ 
	\hline 
    \\
     \scaleto{\bold{0}}{10pt}&  &\scaleto{X_{j+1}(\lambda)}{18pt} & \\
     \\
    \end{array}\right),\quad (j\geq1),
    \end{equation}
with $X_1(\lambda)=X(\lambda)$.\\

The proof of this result is given in two  stages:
\begin{itemize}
    \item[\textbf{Stage 1}]: This stage consists of $L$ steps (where $L$ is the number of $Q_{i,i}$ blocks in $P_A$). In each step, we apply a diagonal argument (either Theorem \ref{thm:zero-case} or Theorem \ref{thm:non-zero-case} of Section \ref{Section:reducible}) starting with the block matrix decomposition $Q_{1,1}$, $Q_{B_1}$ and $P_{2,2}$. For each $i=1,\dots, L$ and for a fixed $\varepsilon>0$, we show that there exists a positive integer $K_{i}>0$ such that the $n_1,n_2,\dots ,n_i$ components of the personalization vector $\bold{v}\in\R^{N\times1}$ under the action of the matrix $X(\lambda)^{K_i}$ in Step $i$ have $1$-norm strictly less than $\varepsilon$. This is done as follows,
    \begin{itemize}
        \item For $i=1,\dots, L-1$, the vector $\bold{v}^TX(\lambda)^{K_i}$ is decomposed into components of size $n_1,n_2\dots,n_{i}$, and finally into two components of size $n_{i+1}$ and $N-n_1-...-n_{i+1}$ for which we apply the resolvent matrix $X_{i+1}(\lambda)$, and
        \item For $i=L$, the vector $\bold{v}^TX(\lambda)^{K_L}$ is decomposed into the first components of size $n_1,n_2\dots,n_{L}$, and finally in a unique component of size $N-n_1-...-n_{L}$, belonging to the dangling part of the network $\mathcal{G}$, where we apply the \textbf{Stage 2}. 
    \end{itemize}
    The idea behind this stage is illustrated in Figure \ref{fig:iterative-argument}.

\begin{figure}[h!]
\centering

\tikzset{every picture/.style={line width=0.75pt}} 

\begin{tikzpicture}[x=0.75pt,y=0.75pt,yscale=-1,xscale=1]

\draw [color={rgb, 255:red, 208; green, 2; blue, 27 }  ,draw opacity=1 ][line width=0.75]    (550,18) .. controls (570,25) and (570,55) .. (550,65) ;
\draw [shift={(545,68)}, rotate = 330] [fill={rgb, 255:red, 208; green, 2; blue, 27 }  ,fill opacity=1 ][line width=0.08]  [draw opacity=0] (8.93,-4.29) -- (0,0) -- (8.93,4.29) -- cycle    ;
\draw [color={rgb, 255:red, 208; green, 2; blue, 27 }  ,draw opacity=1 ][line width=0.75]    (550,18) .. controls (590,32) and (590,95) .. (550,124) ;
\draw [shift={(545,127)}, rotate = 327] [fill={rgb, 255:red, 208; green, 2; blue, 27 }  ,fill opacity=1 ][line width=0.08]  [draw opacity=0] (8.93,-4.29) -- (0,0) -- (8.93,4.29) -- cycle    ;
\draw [color={rgb, 255:red, 208; green, 2; blue, 27 }  ,draw opacity=1 ][line width=0.75]    (550,18) .. controls (610,32) and (610,150) .. (550,195) ;
\draw [shift={(545,198)}, rotate = 325] [fill={rgb, 255:red, 208; green, 2; blue, 27 }  ,fill opacity=1 ][line width=0.08]  [draw opacity=0] (8.93,-4.29) -- (0,0) -- (8.93,4.29) -- cycle    ;

\draw (3.35,6) node [anchor=north west][inner sep=0.75pt]  [font=\small]  {$ \begin{array}{l}
\vspace*{3mm}
\ \ \ \ \ \ \ \ \ \ \ \mathbf{v}^{T} =( \enspace \underbrace{\hspace{4.7mm} \mathbf{v}_{1}^{T} \hspace{4.7mm}}_{{\displaystyle n_{1}}} \enspace | \ \ \underbrace{\hspace{49mm} \mathbf{v}_{2}^{T} \hspace{49mm}}_{\displaystyle N-n_{1}} \enspace )\\
\vspace*{3mm}
\mathbf{v}^{T} X( \lambda )^{K_{1}} =(\enspace\underbrace{\overbracket{\textcolor{white}{sss}\dots\textcolor{white}{sss}}^{\displaystyle\Vert\cdot\Vert_1<\varepsilon} }_{\displaystyle n_1}\enspace | \enspace\underbrace{\textcolor{white}{sss}\dots\textcolor{white}{sss}}_{\displaystyle n_2}\enspace|\enspace\underbrace{\hspace{40mm} \dots \hspace{40mm} }_{\displaystyle N-n_1-n_2}\enspace)\\

\mathbf{v}^{T} X( \lambda )^{K_{2}} =(\enspace\underbrace{\overbracket{\textcolor{white}{sss}\dots\textcolor{white}{sss}}^{\displaystyle\Vert\cdot\Vert_1<\varepsilon} }_{\displaystyle n_1}\enspace |\enspace\underbrace{\overbracket{\textcolor{white}{sss}\dots\textcolor{white}{sss}}^{\displaystyle\Vert\cdot\Vert_1<\varepsilon} }_{\displaystyle n_2}\enspace | \enspace\underbrace{\textcolor{white}{sss}\dots\textcolor{white}{sss}}_{\displaystyle n_3}\enspace|\enspace\underbrace{\hspace{31.1mm}\dots \hspace{31.1mm}}_{\displaystyle N-n_{1} -n_{2} -n_{3}} \enspace )\\
\hspace{5mm}\vdots \hspace{18mm} \vdots \hspace{17mm} \vdots \hspace{17mm} \vdots \hspace{43mm} \vdots \\
\mathbf{v}^{T} X( \lambda )^{K_{L}} =(\enspace\underbrace{\overbracket{\textcolor{white}{sss}\dots\textcolor{white}{sss}}^{\displaystyle\Vert\cdot\Vert_1<\varepsilon} }_{\displaystyle n_1}\enspace |\enspace\underbrace{\overbracket{\textcolor{white}{sss}\dots\textcolor{white}{sss}}^{\displaystyle\Vert\cdot\Vert_1<\varepsilon} }_{\displaystyle n_2}\enspace |\enspace\underbrace{\overbracket{\textcolor{white}{sss}\dots\textcolor{white}{sss}}^{\displaystyle\Vert\cdot\Vert_1<\varepsilon} }_{\displaystyle n_3}\enspace |\enspace\dots\enspace|\enspace\underbrace{\overbracket{\textcolor{white}{sss}\dots\textcolor{white}{sss}}^{\displaystyle\Vert\cdot\Vert_1<\varepsilon} }_{\displaystyle n_{L-1}}\enspace | \enspace\underbrace{\overbracket{\textcolor{white}{sss}\dots\textcolor{white}{sss}}^{\displaystyle\Vert\cdot\Vert_1<\varepsilon} }_{\displaystyle n_{L}}\enspace |\enspace\underbrace{\hspace{9mm}\dots \hspace{9mm}}_{\displaystyle N-n_1-\dots-n_L} \enspace )\\
\end{array}$};

\end{tikzpicture}
\caption{An illustration of the main idea presented in \textbf{Stage 1}.}\label{fig:iterative-argument}
\end{figure}
    \item[\textbf{Stage 2}]: In this stage  we apply the argument given in  Theorem \ref{thm:diagonal-case} to the diagonal part of the matrix $P_A$ located in the last $M$ rows, which corresponds to the dangling clusters of the directed graph $\mathcal{G}$.
\end{itemize}
In what follows, and for simplicity of notation, we will always denote by $I$ the identity matrix with the corresponding order on each case. As it was described above, we proceed with the \textbf{Stage 1} of the proof.
\begin{itemize}
    \item[]\textbf{Step 1 of Stage 1}: Consider the matrix $P_{1,1}=P_A$ which is in terms of $Q_{1,1}$, $Q_{B_1}$, and $P_{2,2}$ in the block decomposition of equation (\ref{eq:block-decomposition}). Regarding $Q_{1,1}$, we have the following two cases:
\begin{enumerate}
    \item[1.1] \textbf{Case $Q_{1,1}$ is a null matrix of order $1$}: From Theorem \ref{thm:zero-case} in Section \ref{Section:reducible} and the recursive equation (\ref{eq:recursive-resolvent}), the resolvent matrix defined by $X_1(\lambda)=X(\lambda)=(1-\lambda)(I-\lambda P_A)^{-1}$ has the form
    \begin{equation*}
    X_1(\lambda)
    =\left(\begin{array}{c|c c c} 
	(1-\lambda)  &  &\lambda Q_{B_1}X_2(\lambda) &\\
    
	\hline 
    \\
    \scaleto{\bold{0}}{10pt}&  &\scaleto{X_2(\lambda)}{18pt} & \\
     \\
\end{array}\right), 
     \end{equation*}   
   where $X_2(\lambda)=(1-\lambda)(I-\lambda P_{2,2})^{-1}$.
   \item[1.2] \textbf{Case $Q_{1,1}$ is a irreducible square matrix of order $n_1$}: From Theorem \ref{thm:non-zero-case} in Section \ref{Section:reducible} and the recursive equation (\ref{eq:recursive-resolvent}), we see that 
    \begin{equation*}
    X_1(\lambda)=\left(\begin{array}{c|c c c}
	(1-\lambda)(I-\lambda Q_{1,1})^{-1}  &  &\lambda (I-\lambda Q_{1,1})^{-1}Q_{B_1}X_2(\lambda)&\\ 
	\hline 
    \\
     \scaleto{\bold{0}}{10pt}&  &\scaleto{X_2(\lambda)}{18pt} & \\
     \\
    \end{array}\right),
    \end{equation*}
   where $X_2(\lambda)=(1-\lambda)(I-\lambda P_{2,2})^{-1}$ and the identity matrix has a different size than in the Case 1.1.
\end{enumerate}

To proceed  with the proof of \textbf{Step 1}, we consider the recursive sequence of vectors $\{\bold{x}_k^T\}_{k\geq0}$ given by 
\begin{equation}\label{eq:sequence-xk-general-case}
 \bold{x}_k^{T}=\bold{x}_{k-1}^{T}X(\lambda)\,,\,(k\geq 1)\qquad\text{with}\qquad \, \bold{x}_0=\bold{y_1} \text{\qquad and\qquad} \Vert \bold{y_1}\Vert_1=1, 
\end{equation}
and let $\bold{y_1}^T=(\enspace(\bold{y_1})_{1}^T\quad|\quad(\bold{y_1})_2^T\enspace)$ be the decomposition of vector $\bold{y_1}$ where $(\bold{y_1})_1\in\R^{n_1\times1}$ and call $\bold{z}_1:=(\bold{y_1})_2\in\R^{(N-n_1)\times1}$ (the second half of the vector $\bold{y}_1$), with $n_1$ the order of the block $Q_{1,1}$. Observe that $n_1=1$ when $Q_{1,1}$ is a null matrix of order $1$. Regardless of what matrix-type is $Q_{1,1}$ (see Section \ref{Section:reducible}),  since $X(\lambda)=X_1(\lambda)$, we see that for a given $\varepsilon>0$, there exits a positive integer $K_1\geq 1$ such that for every $k\geq K_1$ the vector
 \begin{equation*}
 \bold{y_2}^T=\bold{x}_{K_{1}}^TX^{k-K_{1}}(\lambda)=\bold{y_1}^T X_1^k(\lambda)=(\enspace(\bold{y_1})_1^T\quad|\quad(\bold{y_1})_2^T\enspace)X_1^k(\lambda)\qquad\text{with}\qquad \bold{y_2}^T=(\enspace(\bold{y_2})_{1}^T\quad|\quad(\bold{y_2})_2^T\enspace)
 \end{equation*}
satisfies the bound $\Vert (\bold{y_2})_{1}^T \Vert_1<\varepsilon$ with $(\bold{y_2})_1\in\R^{n_1\times1}$. For the next step of \textbf{Stage 1} we will focus  on the vector $\bold{z_2}:=(\bold{y_2})_2\in \R^{(N-n_1)\times 1}$ (the second half of the vector $\bold{y}_2$) and the resolvent matrix $X_2(\lambda)$.
\end{itemize}

The above argument is repeated $L$ times, where on the \textbf{Step j} we consider the decomposition of the vector $\bold{z_j}^T=(\enspace(\bold{z_j})_{1}^T\quad|\quad(\bold{z_j})_2^T\enspace)$ obtained from the \textbf{Step j-1} with $(\bold{z_j})_1\in\R^{n_j\times1}$ and $(\bold{z_j})_2\in\R^{(N-n_1-\,\dots\,-n_j)\times1}$, where $n_j$ is the order of the block $Q_{j,j}$. Therefore, regardless of what matrix-type $Q_{j,j}$ is, for a given $\varepsilon>0$, there exits a positive integer $K_j\geq 1$ with $K_j>K_{j-1}>\dots >K_1$ and such that for every $k\geq K_j$ the vector
 \begin{equation*}
 \bold{y_{j+1}}^T=\bold{z_{j}}^T X_j^k(\lambda)
 =(\enspace(\bold{z_j})_{1}^T\quad|\quad(\bold{z_j})_2^T\enspace)X_j^k(\lambda)\qquad\text{with}\qquad \bold{y_{j+1}}^T=(\enspace(\bold{y_{j+1}})_{1}^T\quad|\quad(\bold{y_{j+1}})_2^T\enspace)
 \end{equation*}
satisfies the bound $\Vert (\bold{y_{j+1}})_{1}^T \Vert_1<\varepsilon$ with $(\bold{y_{j+1}})_1\in\R^{n_j\times1}$.  For the next step of \textbf{Stage 1} we will focus, as above,   on the vector $\bold{z_{j+1}}:=(\bold{y_{j+1}})_2\in \R^{(N-n_1-\,\cdots\,-n_j)\times 1}$ (the second half of the vector $\bold{y_{j+1}}$) and the matrix resolvent matrix $X_{j+1}(\lambda)$.\\

\begin{itemize}
    \item[]\textbf{Step L of Stage 1}: Consider the matrix $P_{L,L}$ which is in terms of $Q_{L,L}$, $Q_{B_L}$, and $P_{L+1,L+1}$ in the block decomposition of equation (\ref{eq:block-decomposition}). Observe that $P_{L+1,L+1}$ is the diagonal matrix corresponding to the dangling cluster part of the directed graph $\mathcal{G}$. We have the following two cases for  $Q_{L,L}$:
\begin{enumerate}
    \item[$L$.1] \textbf{Case $Q_{L,L}$ is a null matrix of order $1$}: From Theorem \ref{thm:zero-case} in Section \ref{Section:reducible} and the recursive equation (\ref{eq:recursive-resolvent}), the resolvent matrix $X_L(\lambda)$ has the form
    \begin{equation*}
    X_L(\lambda)
    =\left(\begin{array}{c|c c c} 
	(1-\lambda)  &  &\lambda Q_{B_L}X_{L+1}(\lambda) &\\ 
	\hline 
    \\
     \scaleto{\bold{0}}{10pt}&  &\scaleto{X_{L+1}(\lambda)}{18pt} & \\
     \\
\end{array}\right), 
     \end{equation*}   
   where $X_{L+1}(\lambda)=(1-\lambda)(I-\lambda P_{L+1,L+1})^{-1}$.
   \item[$L$.2] \textbf{Case $Q_{L,L}$ is a irreducible square matrix of order $n_L$}: From Theorem \ref{thm:non-zero-case} in Section \ref{Section:reducible} and the recursive equation (\ref{eq:recursive-resolvent}), the resolvent matrix $X_L(\lambda)$ has the form
    \begin{equation*}
    X_L(\lambda)=\left(\begin{array}{c|c c c} 
	(1-\lambda)(I-\lambda Q_{L,L})^{-1}  &  &\lambda (I-\lambda Q_{L,L})^{-1}Q_{B_L}X_{L+1}(\lambda)&\\ 
	\hline 
    \\
    \scaleto{\bold{0}}{10pt}&  &\scaleto{X_{L+1}(\lambda)}{18pt} & \\
     \\
    \end{array}\right),
    \end{equation*}
   where $X_{L+1}(\lambda)$  as above and the identity matrix has a different size than in the case $L$.1.
\end{enumerate}
Now, for the vector $\bold{z_L}\in \R^{(N-n_1-\,\cdots\,-n_{L-1})\times 1}$ obtained in \textbf{Step L-1}, we consider the decomposition $\bold{z_L}^T=(\enspace(\bold{z_L})_{1}^T\quad|\quad(\bold{z_L})_2^T\enspace)$ with $(\bold{z_L})_1\in\R^{n_L\times1}$, $(\bold{z_L})_2\in\R^{(N-n_1-\,\dots\,-n_{L-1}-n_L)\times1}$ and $n_L$ is the order of the block $Q_{L,L}$. Therefore, regardless of what matrix $Q_{L,L}$ is, for a given $\varepsilon>0$, there exits a positive integer $K_L\geq 1$ with $K_L>K_{L-1}>\dots >K_1$ and such that for every $k\geq K_L$ the vector
 \begin{equation*}
 \bold{y_{L+1}}^T=\bold{z_{L}}^T X_L^k(\lambda)
 =(\enspace(\bold{z_L})_{1}^T\quad|\quad(\bold{z_L})_2^T\enspace)X_L^k(\lambda)\qquad\text{with}\qquad \bold{y_{L+1}}^T=(\enspace(\bold{y_{L+1}})_{1}^T\quad|\quad(\bold{y_{L+1}})_2^T\enspace)
 \end{equation*}
satisfies the bound $\Vert (\bold{y_{L+1}})_{1}^T \Vert_1<\varepsilon$ with $(\bold{y_{L+1}})_1\in\R^{n_L\times1}$. Now, as before, we focus on the vector $\bold{z_{L+1}}:=(\bold{y_{L+1}})_2\in \R^{(N-n_1-\,\cdots\,-n_{L-1}-n_L)\times 1}$ (the second half of the vector $\bold{y_{L+1}}$) and the resolvent matrix $X_{L+1}(\lambda)$. Now, since the matrix $X_{L+1}(\lambda)=(1-\lambda)(I-\lambda P_{L+1,L+1})^{-1}$ and $P_{L+1,L+1}$ is the diagonal matrix corresponding to the dangling cluster part of the directed graph $\mathcal{G}$, we can move on to \textbf{Stage 2} of the proof.
\end{itemize}

\textbf{Stage 2}: As it was mentioned above, this stage of the proof is based on a diagonal argument corresponding to the dangling cluster part of the directed graph $\mathcal{G}$. Let us denote by $m_j$ the size of each dangling cluster $D_j$, for $j=1,\dots,M$, located in the last rows  $M$ of the matrix $P_A$ in equation (\ref{eq:block-decomposition}).\\

Since $P_{L+1,L+1}=\text{diag}(D_1,\dots,D_M)$ we have
\begin{equation*}
X_{L+1}(\lambda)=(1-\lambda)(I-\lambda P_{L+1,L+1})^{-1}
=\left(\begin{matrix} 
    (1-\lambda)(I-\lambda D_1)^{-1}      & \dots  &\bold{0}\\
    \vdots                               & \ddots & \vdots\\
    \bold{0}                                    &\dots   & (1-\lambda)(I-\lambda D_M)^{-1} 
    \end{matrix}\right)
=\text{diag}(\, Y_1(\lambda),Y_2(\lambda),\dots,Y_M(\lambda)\,),
\end{equation*}
where $Y_j(\lambda)=(1-\lambda)(I-\lambda D_j)^{-1}$, for $j=1,2,\dots, M$. Now, for the vector $\bold{z_{L+1}}\in \R^{(N-n_1-\,\cdots\,-n_{L-1}-n_L)\times 1}$ obtained in \textbf{Step L} of \textbf{Stage 1}, we consider the decomposition 
\begin{equation*}
\bold{z_{L+1}}^T=(\enspace(\bold{z_{L+1}})_{1}^T\quad|\quad(\bold{z_{L+1}})_2^T\quad|\dots|\quad(\bold{z_{L+1}})_M^T\enspace)\qquad \text{with}\qquad  (\bold{z_{L+1}})_1\in\R^{m_1\times1},(\bold{z_{L+1}})_2\in\R^{m_2\times1},\dots, (\bold{z_{L+1}})_M\in\R^{m_M\times1},
\end{equation*}
and $m_1+m_2+\dots+m_M=N-n_1-\,\cdots\,-n_{L-1}-n_L$. Therefore, for every positive integer $k\geq 1$ we have
\begin{equation*}
\bold{z_{L+1}}^TX_{L+1}(\lambda)^k
=(\enspace(\bold{z_{L+1}})_{1}^T\quad|\dots|\quad(\bold{z_{L+1}})_M^T\enspace)
\left(\begin{matrix} 
    Y_1(\lambda)^k  & \dots  &\bold{0}\\
    \vdots          & \ddots & \vdots\\
    \bold{0}               &\dots   & Y_M(\lambda)^k
    \end{matrix}\right)
=(\enspace(\bold{z_{L+1}})_{1}^TY_1(\lambda)^k\quad|\dots|\quad(\bold{z_{L+1}})_M^TY_M(\lambda)^k\enspace).
\end{equation*}
\underline{Claim}

An application of Theorem \ref{thm:diagonal-case}, for a fixed $\varepsilon>0$, gives a positive integer $K_{L+1}\geq 1$ with $K_{L+1}>K_{L}>\dots >K_1$  such that for every $k\geq K_{L+1}$ the inequalities
$$\left\Vert(\bold{z_{L+1}})_{1}^TY_1(\lambda)^k - \alpha_1\bold{c}_1^T\right\Vert_1<\varepsilon,  \qquad 
\left\Vert(\bold{z_{L+1}})_{2}^TY_2(\lambda)^k - \alpha_2\bold{c}_2^T\right\Vert_1<\varepsilon,\qquad\dots,\quad \qquad
\left\Vert(\bold{z_{L+1}})_{M}^TY_M(\lambda)^k - \alpha_M\bold{c}_M^T\right\Vert_1<\varepsilon,$$ 
hold for some $\alpha_1,\alpha_2,\dots, \alpha_M\geq0$ with $\alpha_1+\alpha_2+\dots+\alpha_M=1$,  where $\bold{c}_1\in\R^{m_1\times1}, \bold{c}_2\in\R^{m_2\times1},\dots,\bold{c}_M\in\R^{m_M\times1}$ are the left-hand Perron vectors for the row-normalization of matrices $D_1, D_2,\dots,D_M$, respectively.\\

In order to justify this claim we need first to notice the following. Take a personalization vector $\bold{v}\in\R^{N\times1}$ with $\bold{v}>0$ and $\| \bold{v}\|_1=1$  supported on the nodes belonging to the dangling clusters, that is, all entries of $\bold{v}$ are zero except those belonging to nodes in the dangling clusters. If $\mathcal{G}_j$ denotes the $j$-dangling cluster $(j=1,\dots, M)$,  then call $\alpha_j=\|\bold{v_j}^T\|_1$ where $\bold{v_j}=\bold{v}\chi_{N_j}$ and $N_j$ is the set of nodes belonging to $\mathcal{G}_j$, i.e, $\bold{v}_j$ is the chunk of $\bold{v}$ supported on $\mathcal{G}_j$ (here $\chi_{N_j}$ stands for the indicator function of $N_j$). Evidently 

$$\bold{v}^T X(\lambda)^k=(\enspace \bold{0}\quad|\quad\bold{v_1}^TY_1(\lambda)^k\quad|\dots|\quad\bold{v_M}^TY_M(\lambda)^k\enspace)=(\enspace \bold{0}\quad|\quad\alpha_1\frac{\bold{v_1}^T}{\|\bold{v_1}^T\|_1}Y_1(\lambda)^k\quad|\dots|\quad \alpha_M\frac{\bold{v_M}^T}{\|\bold{v_M}^T\|_1}Y_M(\lambda)^k\enspace)$$

As in (\ref{alphasdiagonal}), where the power method was used, we see that 

$$\bold{v}^T X(\lambda)^k\underset{k}{\longrightarrow} (\underbrace{\enspace \bold{0}_{n_1\times 1}^T\quad|\quad \bold{0}_{n_2\times 1}^T\quad|\dots|\quad \bold{0}_{n_L\times1}^T}_{\text{non-dangling clusters}} \quad|\quad \underbrace{\alpha_1\bold{c}_1^T\,\,\quad |\quad \alpha_2\bold{c}_2^T\,\,\quad|\dots|\quad \alpha_M\bold{c}_M^T\enspace}_{\text{dangling clusters $D_1,\dots, D_M$}}),$$
where $\bold{c}_1\in\R^{m_1\times1}, \bold{c}_2\in\R^{m_2\times1},\dots,\bold{c}_M\in\R^{m_M\times1}$ are the left-hand Perron vectors for matrices $D_1, D_2,\dots, D_M$, respectively.

This will be used now to prove our claim as follows. For a given $\varepsilon>0$ we have proved that there is $K_{L+1}\in\mathbb{N}$ such that for $k\ge K_{L+1}$ we have $\|\ \bold{v}^TX(\lambda)^k\,-\,(\enspace \bold{0}\quad|\quad \bold{z_{L+1}}^T\enspace)\ \|_1<\varepsilon$. 
Hence, since $X(\lambda)$ is row-stochastic we obtain (see Remark \ref{rmk:operator-norm})
$$\|\ \bold{v}^TX(\lambda)^{k+j}\,-\,(\quad \bold{0}\quad|\quad \bold{z_{L+1}}^TX(\lambda)^j\quad)\ \|_1<\varepsilon,$$ for all $j$. 
In order to have a vector of norm one we just need to  $\varepsilon$-modify the second half of $(\quad\bold{0}\quad|\quad \bold{z_{L+1}}^T\enspace)$ 
into $(\quad \bold{0}\quad|\quad {\bold{\tilde{z}_{L+1}}}^{\quad T}\enspace)$ so that $\|\ (\quad \bold{0}\quad|\quad \bold{\tilde{z}_{L+1}}^{\quad T}\enspace)\ \|_1=1$ and  $\|\enspace\bold{\tilde{z}_{L+1}}^{\quad T}-\bold{{z}_{L+1}}^{T}\ \|_1<\varepsilon$. From the remarks above follows

$$(\enspace \bold 0\quad|\quad \bold{\tilde{z}_{L+1}}^{\quad T}\ ) X(\lambda)^j\underset{j}{\longrightarrow} (\underbrace{\enspace \bold{0}_{n_1\times 1}^T\quad|\quad \bold{0}_{n_2\times 1}^T\quad|\dots|\quad \bold{0}_{n_L\times1}^T}_{\text{non-dangling clusters}} \quad|\quad \underbrace{\alpha_1\bold{c}_1^T\,\,\quad |\quad \alpha_2\bold{c}_2^T\,\,\quad|\dots|\quad \alpha_M\bold{c}_M^T\enspace}_{\text{dangling clusters $D_1,\dots, D_M$}}),$$
and thus the same happens with   $(\enspace \bold{0}\quad|\quad \bold{z_{L+1}}^T\ )X(\lambda)^j$ and $\bold{v}^TX(\lambda)^{k+j}$.





All the previous work can be now put  together 
to conclude that the sequence $\{\bold{x}_k\}_{k\geq0}$ in equation (\ref{eq:sequence-xk-general-case}) converges to a non-zero vector $\bold{x}\in\R^{N\times 1}$ with $\Vert\bold{x}\Vert_1=1$ such that
\begin{align*}
    \bold{x}^T=\lim_{k\to\infty}\bold{v}^TX^{k}(\lambda)
    =\lim_{k\to\infty}
    (\underbrace{\enspace (\bold{x}_k)_1^T\quad|\quad (\bold{x}_k)_2^T \quad|\dots|\quad (\bold{x}_k)_L^T}_{\text{non-dangling clusters}}\quad&|\quad \underbrace{(\bold{x}_k)_{L+1}^T\quad|\quad (\bold{x}_k)_{L+2}^T \quad|\dots|\quad (\bold{x}_k)_{L+M}^T\enspace}_{\text{dangling cluster $D_1,\dots, D_M$}})\\
    =(\underbrace{\enspace \bold{0}_{n_1\times1}^T\quad|\quad \bold{0}_{n_2\times1}^T\quad|\dots|\quad \bold{0}_{n_L\times1}^T}_{\text{non-dangling clusters}} \quad&|\quad \underbrace{\alpha_1\bold{c}_1^T\,\,\quad |\quad \alpha_2\bold{c}_2^T\,\,\quad|\dots|\quad \alpha_M\bold{c}_M^T\enspace}_{\text{dangling clusters $D_1,\dots, D_M$}}),
\end{align*}
for some $\alpha_1,\alpha_2,\dots, \alpha_M\geq0$ with $\alpha_1+\alpha_2+\dots+\alpha_M=1$ and where $\bold{c}_1\in\R^{m_1\times1}, \bold{c}_2\in\R^{m_2\times1},\dots,\bold{c}_M\in\R^{m_M\times1}$ are the left-hand Perron vectors for matrices $D_1, D_2,\dots, D_M$, respectively.\\

\begin{rmk}
After close inspection we see that $(\bold{x}_{k})_{L+j}^TY_j(\lambda)\le (\bold{x}_{k+1})_{L+j}^T$ which implies  $\|(\bold{x}_k)_{L+j}^T\|_1\le\|(\bold{x}_{k+1})_{L+j}^T\|_1$ for  $j=1,\dots, M$. This together with $\|(\bold{x}_k)_{L+j}^T\|_1\le 1$ for all $k$ and $j=1,\dots, M$ also justifies the existence of $\alpha_j\in\R$ with $0\leq\alpha_j\leq 1$ satisfying $\displaystyle\lim_{k\to\infty}\|(\bold{x}_{k})_{L+j}^T\|_1=\alpha_j$.
\end{rmk}
\end{proof}

 We make the important although evident remark that for every arbitrary choice of $(\alpha_1,\dots,\alpha _M)$ there is a personalization vector $\bold{v}\in\R^{N\times1}$ as above supported on the nodes of the dangling clusters  satisfying $\alpha_i=\|\bold{v_i}\|_1$ for all $i=1,\dots, M$ (see ($iv$) in Corollary \ref{cor:genral-case} below). \\

Similarly to the previous cases, Theorem~\ref{thm:genral-case} can be reinterpreted to completely solve the problem of existence and uniqueness of fixed points of the PageRank for a general non-strongly connected network $\mathcal{G}=(V,E)$ with  a given damping factor $\lambda\in (0,1)$. More precisely,\\
 \begin{coro}\label{cor:genral-case}
Let $\mathcal{G}=(V,E)$ be a non-strongly connected graph with $M\ge 1$ dangling clusters. Then, 
\begin{itemize}
 \item[{\it (i)}] $PR_\lambda: \Delta_N^+\longrightarrow \Delta_N^+$ has some fixed point if and only if $\mathcal{G}$ has no non-dangling clusters.
 \item[{\it (ii)}]  If $\mathcal{G}$ has some non-dangling clusters, then  $PR_\lambda: \Delta_N\longrightarrow \Delta_N$ has some fixed points.
 \item[{\it (iii)}] $PR_\lambda$ has a unique fixed point if and only if $\mathcal{G}$ has only one dangling cluster.
 \item[{\it (iv)}] If $\mathcal{G}$ has $M>1$ dangling clusters, then the set of fixed points of $PR_\lambda$ is the convex hull  in $\Delta_N^+$ (or $\Delta_N$) of the left-hand Perron vectors of the matrices the row-normalized and irreducible of matrices $D_1, D_2,\dots, D_M$ given in Theorem~\ref{thm:genral-case}.
\end{itemize}
\end{coro}

\section*{Acknowledgement}

This work has been supported by INCIBE/URJC Agreement M3386/2024/0031/001 within the framework of the Recovery, Transformation and Resilience Plan funds of the European Union (Next Generation EU) and by project M3707 (URJC Grant).












\end{document}